\documentclass[journal]{IEEEtran}

\usepackage{epstopdf}

\usepackage{setspace}
\usepackage{amsmath}
\usepackage{amssymb}
\usepackage{amsthm}
\usepackage{dsfont}
\usepackage{multirow}
\usepackage{amsfonts}
\usepackage{color}
\usepackage{graphicx}
\usepackage{tikz}
\usetikzlibrary{fit,calc,positioning,decorations.pathreplacing,matrix}
\usetikzlibrary{shapes,arrows}
\usepackage{subcaption}
\usepackage[]{algorithmicx}
\usepackage{algpseudocode,algorithm}
\usepackage{cite}
\usepackage{mathtools}
\usepackage{stmaryrd}
\usepackage[mathscr]{euscript}
\usepackage{mathrsfs}

\usepackage{soul}

\newtheorem{theorem}{Theorem}

\newcommand{\x}{\mathbf{x}}
\newcommand{\y}{\mathbf{y}}
\newcommand{\z}{\mathbf{z}}

\newcommand{\s}{\mathbf{s}}
\newcommand{\h}{\mathbf{h}}
\newcommand{\bvec}{\mathbf{b}}
\newcommand{\ve}{\mathbf{v}}
\newcommand{\proj}{\mathbf{\Phi}}

\newcommand{\Cs}{\mathbf{C}_{\mathbf{s}_0}}
\newcommand{\W}{\mathbf{W}}

\DeclareMathOperator*{\argmin}{arg\,min}

\author{ Bahareh Tolooshams$^{*}$, {\it Student Member, IEEE}, Satish~Mulleti$^{*}$, {\it Member, IEEE}, Demba Ba, and Yonina C. Eldar, {\it Fellow, IEEE}
	\thanks{\scriptsize B. Tolooshams and D. Ba are with the School of Engineering and Applied Sciences, Harvard University, Cambridge, MA; S. Mulleti is with the Department of Electrical Engineering, Indian Institute of Technology (IIT) Bombay, Mumbai, India; Y. C. Eldar is with the Faculty of Math and Computer Science, Weizmann Institute of Science, Israel. Email: btolooshams@seas.harvard.edu, mulleti.satish@gmail.com, demba@seas.harvard.edu, yonina.eldar@weizmann.ac.il}
	\thanks{\scriptsize {$^{*}$Equal contribution.
			B. Tolooshams and D. Ba acknowledge the partial support of NSF under grant number DMS-2134157 and the ARO under grant number W911NF-16-1-0368. The latter is part of a collaboration between US DOD, UK MOD and UK Engineering and Physical Research Council (EPSRC) under the Multidisciplinary University Research Initiative. For S. Mulleti and Y. C. Eldar, this research was partially supported by the Israeli Council for Higher Education (CHE) via the Weizmann Data Science Research Center, by a research grant from the Estate of Tully and Michele Plesser, by the European Union’s Horizon 2020 research and innovation program under grant No. 646804-ERC-COG-BNYQ, by the Israel Science Foundation under grant no. 0100101, and by the QuantERA grant  C’MON-QSENS.}}
		\thanks{Manuscript accepted May, 2023. Preliminary results of this work appeared in the proceedings of IEEE Int. Conf. Acoust., Speech and Signal Process. (ICASSP) 2021.}
}

\definecolor{myblue}{RGB}{0,0, 120}
\usepackage{hyperref}
\hypersetup{
	colorlinks=true,
	linkcolor=myblue,
	citecolor=myblue,
}
\usepackage[nameinlink,capitalise]{cleveref}
\usepackage[numbers]{natbib}

\begin{document}
\title{Unrolled Compressed Blind-Deconvolution}

	\markboth{Accepted to the IEEE Transactions of Signal Processing}%
{Shell \MakeLowercase{\textit{et al.}}: Bare Demo of IEEEtran.cls for Journals}
\maketitle
		
 	\begin{abstract}
The problem of sparse multichannel blind deconvolution (S-MBD) arises frequently in many engineering applications such as radar/sonar/ultrasound imaging. To reduce its computational and implementation cost, we propose a compression method that enables blind recovery from much fewer measurements with respect to the full received signal in time. The proposed compression measures the signal through a filter followed by a subsampling, allowing for a significant reduction in implementation cost. We derive theoretical guarantees for the identifiability and recovery of a sparse filter from compressed measurements. Our results allow for the design of a wide class of compression filters. We, then, propose a data-driven unrolled learning framework to learn the compression filter and solve the S-MBD problem. The encoder is a recurrent inference network that maps compressed measurements into an estimate of sparse filters. We demonstrate that our unrolled learning method is more robust to choices of source shapes and has better recovery performance compared to optimization-based methods. Finally, in data-limited applications (fewshot learning), we highlight the superior generalization capability of unrolled learning compared to conventional deep learning.
 	\end{abstract}
	
	\IEEEpeerreviewmaketitle
	
\section{Introduction}

In a variety of remote sensing applications (e.g., radar imaging~\cite{herman2009high}, seismic signal processing~\cite{filho201seismic}, sonar imaging~\cite{carter_sonar}, and ultrasound imaging~\cite{eldar_sos,eldar_beamforming}), an unknown source signal reflected from sparsely located targets is measured through a multichannel receiver system. The signal received at each channel can be modelled as a linear convolution of a source signal and sparse filters; the filters, modelling the targets' locations, are also known as impulse responses of the channels~\cite{bajwa_radar, bar_radar}. Sparse multichannel blind-deconvolution (S-MBD) is the problem of recovering the source signal shape and sparse filters from received measurements~\cite{wang_chi, kazemi2014sparse}.

Bilen et al. \cite{bilen} proposed an $\ell_1$-based convex optimization method to solve the S-MBD problem. Gribonval et al. \cite{gribonval_dl} and Mulleti et al. \cite{mulleti_mbd} studied S-MBD in the context of sparse dictionary calibration, and Garcia-Cardona and Wohlberg~\cite{garcia_cdl} formulated it via convolutional dictionary learning. Theoretical analysis on identifiability of S-MBD exists in the literature~\cite{lee_17, wang_chi, cosse_mbd}. These works derived the identifiability results in terms of the number of channels by assuming that the sparse filters are random and all the measurements in each channel are available. Specifically, they showed that given $M$ measurements in each channel, $N = \mathcal{O}(M \log^4 M)$ channels are sufficient for identifiability \cite{lee_17, wang_chi}. In addition to using a large number of channels, these results rely on having access to full measurements at each channel. None of these methods uses the sparsity of the filters to reduce the number of measurements in each channel. In a multichannel receiver system, hardware cost and power consumption increase with the number of channels, as each channel requires a dedicated analog-to-digital converter and corresponding circuitry. Further, the computational cost depends on the number of channels and the number of measurements per channel, as a larger number of measurements requires more storage and computations. To reduce the hardware implementation and computational cost, it is desired to solve S-MBD from compressed (few) measurements captured from fewer channels.

In the context of dictionary learning, Tolooshams et al.~\cite{tolooshams2018mlsp,tolooshams2020tnnls} proposed structured inference neural networks to solve the S-MBD problem using algorithm unrolling~\cite{monga2019algorithm, shlezinger2022model}; however, their method still relies on full measurements. Chang et al. \cite{chang2019randnet} further extend the unrolling framework to enable dictionary learning from compressed measurements through an unstructured random matrix; this compression has two drawbacks. First, random-matrix-based compression may not achieve large compression rates in practice. Second, their practical implementation may be challenging due to their unstructured nature.

Mulleti et al. \cite{mulleti_mbd} addressed the later drawback and proposed a structured, linear compression operator that could be realized using a finite impulse response (FIR) filter. The authors took a frequency-domain approach for S-MBD and derived the identifiability results by using a fixed sum-of-sincs (SoS) FIR filter for compression. In addition, the authors applied the blind dictionary calibration (BDC) method \cite{gribonval_dl} and the truncated-power iteration (TPI) approach \cite{li2017blind} to recover sparse filters from the compressed Fourier measurements. This approach's limitation is that the source's frequency spectrum is assumed not to vanish on selected points. Since the source is unknown in advance, designing the filter is challenging. Furthermore, the structure of the filter is fixed and independent of the reconstruction method.

Recent work addressed this challenge by offering a data-dependent joint design of the frequency locations of a continuous-time SoS filter and the reconstruction parameters~\cite{mulleti_los_icassp}. However, the filters still follow the SoS structure with learnable frequency locations, resulting in limited flexibility. Hence, a natural question is whether a general condition on the filters, which does not require an SoS form, can be derived such that compression is possible with perfect recovery. Moreover, can we leverage data to learn such a filter together with the reconstruction method?

\subsection{Contributions}

In this paper, we focus on answering the preceding two questions. Starting from a time-domain approach, we consider the problem of filter-based compression for S-MBD. We derive time-domain identifiability conditions on the compression filter and the S-MBD measurements to enable recovery. The time-domain approach helps in deriving general conditions on the filters compared to the frequency-domain framework. Moreover, we offer learning-based methods for designing the filters which were briefly proposed and discussed in~\cite{tolooshams2021lsmbd}. Our main contributions are summarized as follows.
 
\begin{itemize}
    \item We formulate a compressive sparse-MBD problem based on a hardware-efficient linear compression operator that can be realized using an FIR filter. Unlike prior theoretical analysis on Fourier samples~\cite{mulleti_mbd}, we provide time-domain theoretical guarantees for the recovery of sparse filters from reduced measurements. We derive identifiability conditions jointly on the compression filter and the source in terms of spark properties on the convolution matrices of the source and compression filter. Unlike non-vanishing frequency requirements of the source's Fourier spectrum in prior work~\cite{mulleti_mbd}, our framework neither enforces any specific structural requirement on the compression filter nor places vanishing spectrum conditions on the source. However, the above-mentioned flexibility on identifiability conditions comes with a trade-off, a need for more measurements than the minimum required Fourier samples achieved by~\cite{mulleti_mbd} for identifiability.  
    \item We propose unrolled inference networks (\cref{fig:infmap_unroll}); this is to give an unrolling structure to standard deep learning inference networks (\cref{fig:infmap_nn}) using our compressive S-MBD problem formulation. We build upon our preliminary results in~\cite{tolooshams2021lsmbd} and enable learning of the compression filter from data. We note that the architecture shown in \cref{fig:ls-mbd}a is previously proposed by~\cite{tolooshams2021lsmbd}. We show that unrolled learning outperforms optimization-based methods and highlight its robustness to the source shapes of Gaussian and symmetric wavelet (\cref{fig:fsmbdbaseline}). We characterize the effect of unrolling, as compression ratio (CR) varies, e.g., only a few unrolled layers are needed to achieve highly accurate identification of the target locations (\cref{fig:unroll}). Furthermore, we propose an additional linear compression, performing filtering followed by decimation; the compression is more memory efficient (\cref{tab:memory}) and has a minimal decline in performance (\cref{fig:bank}). Our results differ from~\cite{tolooshams2021lsmbd} as follows: \cite{tolooshams2021lsmbd} a) provided no theoretical analysis of their proposed unrolled network, b) considered only the noiseless setting where the source shape is Gaussian, hence, no robust characterization was performed, and c) focused on target amplitude recovery rather than hit rate. Moreover, the new compression based on decimated filtering and the following contribution are unique to this paper.
    \item We propose combining a network denoiser with an unrolled network and show that this framework significantly improve performance in the task of recovery from compressed noisy measurements (\cref{fig:noisy}). Finally, we highlight the advantage of unrolled learning over conventional deep learning in fewshot learning; our framework is significantly superior to generic deep networks in the low data regime (\cref{fig:fewshot}).
\end{itemize}

\subsection{Notations and Organizations}
We use the following notation in the paper. Vectors and matrices are denoted by lower and uppercase boldfaced letters, respectively. A real-valued, $N$-length, causal sequence $\{s[0], s[1], \cdots, s[N-1]\}$ is denoted by a vector $\mathbf{s} = [s[0], s[1], \cdots, s[N-1]]^{\mathrm{T}} \in \mathbb{R}^N$  where the superscript $\mathrm{T}$ denotes transpose. For any two sequences $\mathbf{s} \in \mathbb{R}^{M_s}$ and $\mathbf{x} \in \mathbb{R}^{M_x}$, their convolution is $\mathbf{s}*\mathbf{x}\in \mathbb{R}^{M_s+M_x-1}$. For any vector $\mathbf{s} \in \mathbb{R}^{M_s}$, a convolution matrix with $M$ columns is constructed as,
\begin{align}
	\mathbf{C}_s(M) =  \begin{pmatrix}
		\mathbf{s}[0] &   & 0 \\
		\mathbf{s}[1] & \ddots  & \vdots \\
		\vdots &  \ddots & \mathbf{s}[0] \\
		\mathbf{s}[M_s-1] &   & \mathbf{s}[1]\\
		\vdots  &  \ddots &  \vdots \\
		0 &   &  \mathbf{s}[M_s-1]         
	\end{pmatrix}.
\end{align}
With this notation, the convolution $\mathbf{s}*\mathbf{x}$ is represented as $\mathbf{C}_s(M_x) \mathbf{x} = \mathbf{C}_x(M_s) \mathbf{s}$. The symbol $\underset{t}{\ast}$ denotes truncated convolution, with the longer vector appearing on the left side of the symbol. For $M_s \geq M_x$, $\mathbf{s}\underset{t}{\ast} \mathbf{x}$ denotes values $\mathbf{s}*\mathbf{x}[m]$ for $M_x \leq m \leq M_s$. The spark of a matrix $\mathbf{A}$, that is, $\text{Spark}(\mathbf{A})$ is defined as the minimum number of columns of $\mathbf{A}$ that are linearly dependent. 

The paper is organized as follows. In~\cref{sec:pf}, we formulate the S-MBD problem mathematically and present our assumptions. Time-domain identifiability results are derived in~\cref{sec:identifiability}. In~\cref{sec:unroll}, we present a data-driven unrolled-based compression filter design; we outline various unrolled architectures and discuss their training mechanism. Numerical results in~\cref{sec:results} followed by a conclusion.

\section{Problem Formulation}
\label{sec:pf}
Consider a set of $N$ signals given as,
\begin{align}
    \mathbf{y}^n = \mathbf{s}_0*\mathbf{x}^n, \quad n =1,\ldots, N,
    \label{eq:mbd1}
\end{align}
where $*$ denotes linear convolution operation. In this model, $\mathbf{s}_0$ represents a common source signal, and $\mathbf{x}^n$ denotes $n$-th channel filter. 

The problem of multichannel blind-deconvolution (MBD) is to determine the source and filters from the measurements $\{\mathbf{y}^n\}_{n=1}^N$. The problem is ill-posed without further assumptions on either the source, or the filters, or both. Sparse MBD (S-MBD) is one such widely applicable model where the filters are assumed to be sparse. In addition, the source is assumed to have a finite impulse response (FIR). Specifically, we make the following structural assumptions:
\begin{enumerate}
    \item [(A1)] \textbf{Sparse Filters}$:\mathbf{x}^n\in \mathbb{R}^{M_x}$ and $\|\mathbf{x}^n\|_0\leq L$ for $n = 1, \ldots, N.$
    \item [(A2)] \textbf{FIR Source}$:\mathbf{s}_0\in \mathbb{R}^{M_{s_0}}$.\\
    \item [(A3)] \textbf{Sparse Coprime Filters}: Consider the decomposition $\mathbf{x}^n = \mathbf{h}_0*\bar{\mathbf{x}}^n$ for all $n$. If length of $\mathbf{h}_0$ is greater than one, then $\|\bar{\mathbf{x}}^n\|_0 > L$ for all $n$.
    \item [(A3$^{\prime}$)] \textbf{Non-zero first and last elements of the source}: $\mathbf{s}_0[1] \neq 0$ and $\mathbf{s}_0[M_{s_0}] \neq 0$.
    
\end{enumerate}

The sparse filter model with FIR source is ubiquitous in  applications \cite{mulleti_mbd}. For example, in radar imaging with $N$ receivers, $\mathbf{s}_0$ denotes the transmit pulse. By assuming that the target scene consists of $L$ stationary point targets, the signal received by the $n$-th receiver can be modelled as $\mathbf{y}^n$ in \eqref{eq:mbd1}. Here, non-zero locations and amplitudes of the sparse filter $\mathbf{x}^n$ contain the information of the distances and amplitudes of the targets for the $n$-th receiver. 
As one can generate only compactly supported transmits signals, the FIR source model is justified. As discussed in \cite{mulleti_mbd}, either (A3) or (A3$^\prime$) ensure unique identifiability \footnote{Assumption (A3$^\prime$) is simple to verify in practice compared to assumption (A3). For identifiability, we need either (A3) or (A3$^\prime$) to be satisfied.} Here, identifiability implies that the source and sparse filters can be uniquely determined up to a scaling constant. As we are representing one-dimensional sequences as finite-length vectors with known support, we do not have shift ambiguity which alters the support of the sequences. 

In the compressive S-MBD problem, the objective is to determine either the sparse filters or the source or both from compressed measurements of $\{\mathbf{y}^n\}_{n=1}^N$. In this paper, we focus on recovering the sparse filters $\{\mathbf{x}^n\}_{n=1}^N$s from compressed measurements. This problem is equivalent to estimating the locations of the targets in radar~\cite{herman2009high} or ultrasound imaging~\cite{eldar_sos,eldar_beamforming}. The desirable properties of the compression operator are that it should be linear, practically implementable, and importantly allows perfect recovery with the highest compression rate possible. To this end, we use linear-filtering-based compression. Specifically, to compress the measurements, we consider subsampling the filtered signal
\begin{align}
\mathbf{z}^n = \mathbf{h}*\mathbf{y}^n = \mathbf{h}*\mathbf{s}_0*\mathbf{x}^n,
\label{eq:z}
\end{align}    
where $\mathbf{h}\in \mathbb{R}^{M_h}$ is a compression filter applied across all channels. Note that the compression mechanism is similar to that in \cite{mulleti_mbd}. However, a major difference is that in our framework, the filter $\mathbf{h}$ does not require a discrete-SoS (D-SoS) model. The compressed measurements can be written as
\begin{align}
\mathbf{z}^n =  \mathbf{s}*\mathbf{x}^n,
\label{eq:z1}
\end{align}
where $\mathbf{s} = \mathbf{h}*\mathbf{s}_0 \in \mathbb{R}^{M_s}$ is the \emph{effective source} of length $M_{s} = M_{s_0}+M_h-1$, and $M_y = M_{s_0}+M_x-1$. If the number of subsampled measurements of $\mathbf{z}^n$ is less than $M_y$ then compression is achieved.

In the following section, we discuss conditions on the filter $\mathbf{h}$, or more specifically on the effective source $\mathbf{s}$ to ensure identifiability up to a scaling ambiguity. Specifically, if there exist an alternative set of effective source $\hat{\mathbf{s}}$ and sparse filters $ \{\hat{\mathbf{x}}^n\}_{n=1}^N$, such that 
\begin{align}
\mathbf{z}^n =  \hat{\mathbf{s}}*\hat{\mathbf{x}}^n,\quad n = 1, \ldots, N,
\label{eq:alt_soln}
\end{align}
on the subsampled locations then, there exists a scalar $\beta \neq 0$ such that
\begin{align}
\hat{\mathbf{s}} = \beta \, \mathbf{s} \quad \text{and} \quad \hat{\mathbf{x}}^n = \mathbf{x}^n/\beta, \quad n=1,\ldots, N. \label{eq:scale_amb}
\end{align}
If a set of alternative solutions satisfies \eqref{eq:alt_soln} but not \eqref{eq:scale_amb} then the problem is not uniquely identifiable.

\section{Time-Domain Theoretical Guarantees For Compressed S-MBD}
\label{sec:identifiability}
\begin{figure}[!t]
	\centering
	\begin{subfigure}[t]{0.99\linewidth}
		\includegraphics[width = 3.4 in]
		{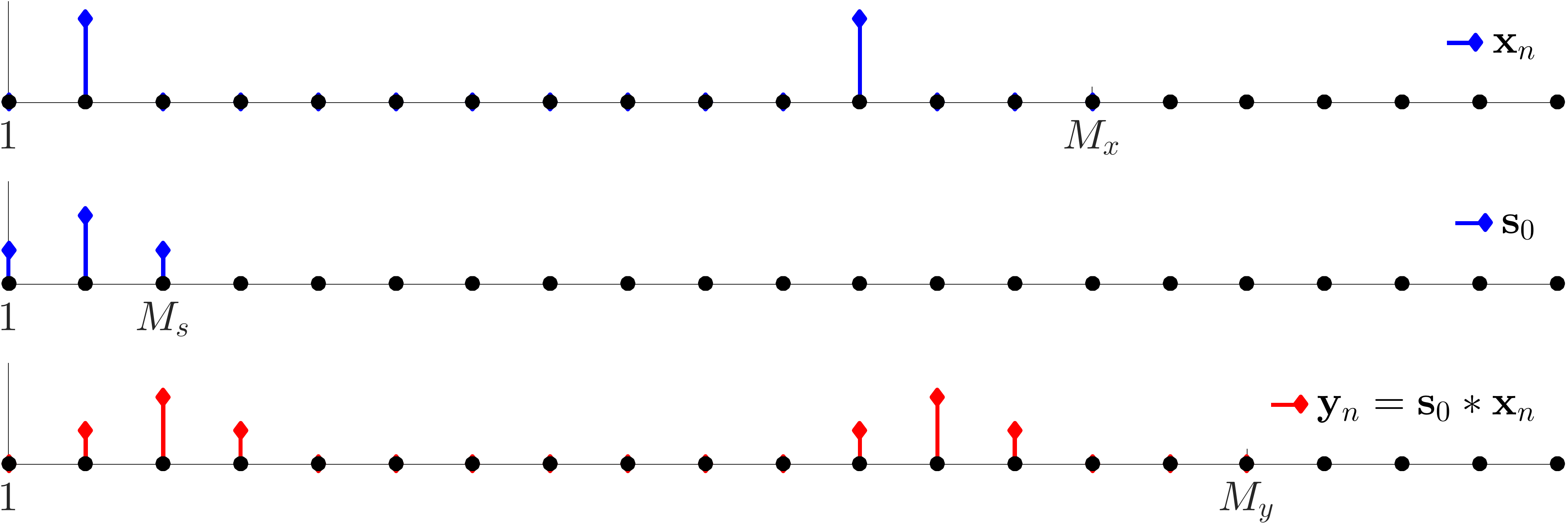}
		\caption{MBD measurement.}
		\label{fig:example1}
		\end{subfigure}
		\begin{subfigure}[t]{0.99\linewidth}
 		\includegraphics[width = 3.4in]{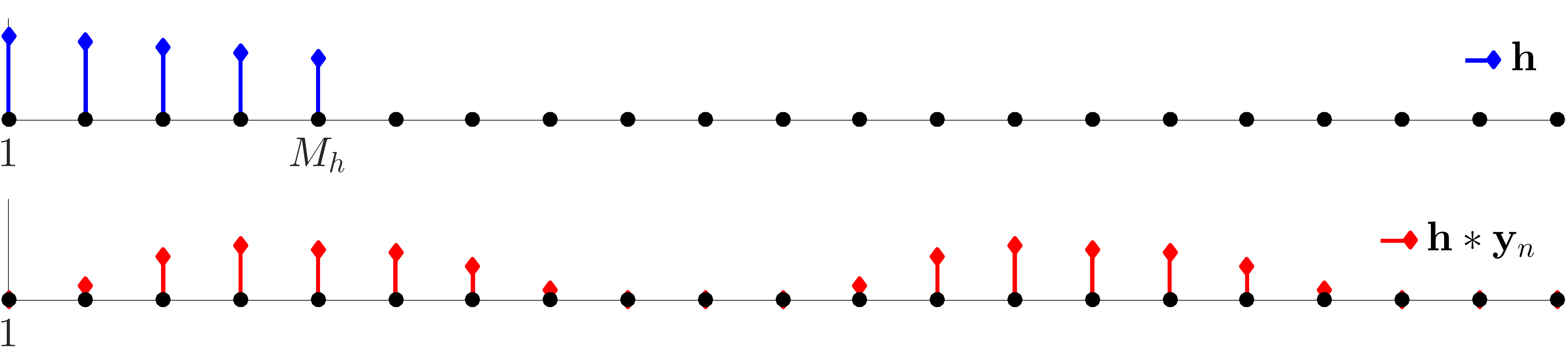}
 		\caption{Filtering by a short filter.}
  		\label{fig:example2}
		\end{subfigure}
		\begin{subfigure}[t]{0.99\linewidth}
		\includegraphics[width = 3.4in]{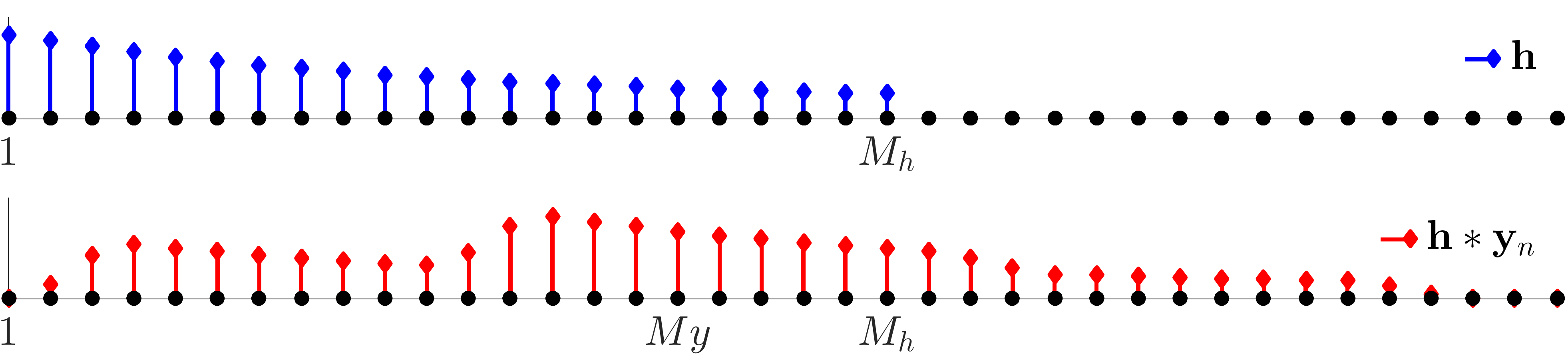}
		\caption{Filtering by a long filter.}
		\label{fig:example3}
		\end{subfigure}
\caption{A toy example to demonstrate that recovery from compression through subsampling may not be possible for $M_h \leq M_y$.}
\label{fig:example}
\end{figure}

In this section, we discuss the problem of uniquely identifying the sparse filters $\{\mathbf{x}^n\}_{n=1}^N$ from subsamples of $\{\mathbf{z}^n\}_{n=1}^N$. The identifiability depends on the subsampling strategy, which in turn depends on the relative lengths of the compression filter $\mathbf{h}$, source $\mathbf{s}_0$, and sparse filters. 

The MBD measurements in \eqref{eq:z1} can be written as
\begin{align}
    \mathbf{z}^n = \mathbf{C}_{\h} \, \Cs \mathbf{x}^n \in \mathbb{R}^{M_z}, \quad n = 1, \ldots, N,
\label{eq:mbd_matrix}
\end{align}
where $\Cs \in \mathbb{R}^{M_y \times M_x}$  and $\mathbf{C}_{\h} \in \mathbb{R}^{M_z \times M_y}$ are convolution matrices corresponding to the source and filter,  respectively. To identify $\{\mathbf{x}^n\}_{n=1}^N$ from subsamples of $\{\mathbf{z}^n\}_{n=1}^N$, we consider a binary-compression or subsampling operator $\mathbf{B} \in \{0, 1\}^{K \times M_z}$ where $K<M_z$ is the number of compressed measurements. Each row of $\mathbf{B}$ has exactly one non-zero element, and the matrix has a full row rank. The compressed measurements are given as
\begin{align}
    \bar{\mathbf{z}}^n = \mathbf{B}\mathbf{z}^n= \mathbf{B} \mathbf{C}_{\h}\, \Cs \mathbf{x}^n, \quad n = 1, \ldots, N.
    \label{eq:compressed_mes}
\end{align}
The identifiability depends on the number of subsamples $K$, subsampling pattern or matrix $\mathbf{B}$, and properties of $\mathbf{C}_{\h}\, \Cs$. 

To proceed further, we present a few basic properties of $\mathbf{C}_{\s}  = \mathbf{C}_{\h}\, \Cs $ that will allow subsampling. We consider two possible scenarios: (i) $M_h \leq M_y$ and (ii) $M_h> M_y$. We first show that recovery from compression through subsampling may not be possible for $M_h \leq M_y$. Consider a toy example with $M_x = 15$, $L=2$, and $M_{s_0} = 3$ (See~\cref{fig:example1}). In~\cref{fig:example2}, we chose a compression filter of length $M_h = 4$ and plot the filtered signal $\mathbf{h}*\mathbf{y}^n$. We observe that a few values of $\mathbf{z}^n = \mathbf{h}*\mathbf{y}^n$ are zero due to the short compression filter. Subsampling those values will result in a loss of information. In addition, each non-zero sample of $\mathbf{z}^n$ only has a contribution from a few non-zero entries of $\mathbf{x}^n$. This implies that the information of $\mathbf{x}^n$ is not well spread across the samples of $\mathbf{z}^n$ which is desirable for compression. 

Next, consider a long compression filter such that $M_h> M_y$. In the toy example, we set $M_h = M_y + 5$. The compression filter and $\mathbf{z}^n = \mathbf{h}*\mathbf{y}^n$ are shown in~\cref{fig:example3}. We note that none of the values are zero unless the convolution sum is zero in that sample. Further, it can be verified that the values of $\mathbf{z}^n[m]$ for $M_y-1 \leq m \leq M_h-1$ have a contribution from all the samples of $\mathbf{y}^n$ \cite{tolooshams2021lsmbd}. This implies that these samples, indexed by the set $\mathcal{K}= \{M_y-1, M_y, \cdots, M_h-1\}$, carry information of all the non-zero components of the sparse filter. The set $\mathcal{K}$ includes the sample of the convolution $\mathbf{C}_{\h}$ where $\h$ and $\y$ are fully overlapped. Hence, perfect recovery through subsampling may be possible by considering the measurements $\{\mathbf{z}^n[m]\}_{m \in \mathcal{K}}$. Given this intuition, we set the operator $\mathbf{B}$ to subsample the measurements from the set $\mathcal{K}$ where the compression ratio is controlled by the relative length of the compression filter $M_h$ and signal $M_y$. For the remainder of this section, we denote $\mathbf{T}_h(r) \coloneqq \mathbf{B} \mathbf{C}_{\h}$, applied on a signal of length $r$. For a fixed $M_y$, the number of compressed measurements $K = |\mathcal{K}| = M_h-M_y+1$ increases linearly with the length of the compression filter. The measurements $\{\mathbf{z}^n[m]\}_{m \in \mathcal{K}}$ are said to be compressed if $K<M_y$. We next derive conditions that ensure compression together with identifiability. Given that the subsamples $\{\mathbf{z}^n[m]\}_{m\in \mathcal{K}}$ carry information about the sparse filters, the questions required to be answered are whether the sparse filters can be identifiable from these measurements with the assumption that $K<M_y$. To answer the question, we next study the problem of identifying the sparse filters from the subsamples of the truncated convolution
\begin{align}
	\bar{\mathbf{z}}^n = \mathbf{T}_h(M_y) \mathbf{y}^n = \mathbf{h} \underset{t}{\ast} \mathbf{y}^n, 
	\label{eq:subsampled}
\end{align}
where the Toeplitz matrix $\mathbf{T}_h(r) \in \mathbb{C}^{ (M_h-r+1)\times r}$, for $1 \leq r \leq M_h$, is constructed as
\begin{align}
	\mathbf{T}_h(r) = 
	\begin{pmatrix}
		\mathbf{h}[r-1] & \mathbf{h}[r-2] &  \cdots & \mathbf{h}[0] \\
		\mathbf{h}[r] & \mathbf{h}[r-1] &  \cdots & \mathbf{h}[1] \\
		\mathbf{h}[r+1] & \mathbf{h}[r] &  \cdots & \mathbf{h}[2] \\
		\vdots  & \vdots  & \ddots & \vdots   \\
		\mathbf{h}[M_h-1] & \mathbf{h}[M_h-2] &  \cdots &  \mathbf{h}[M_h-r] \\ 
	\end{pmatrix}.
	\label{eq:Tmat}
\end{align}
For $r = M_y$, the matrix consists of $K = M_h-M_y+1$ rows or $K$ compressed measurements. The compressed measurements can alternatively be written as 
\begin{align}
\bar{\mathbf{z}}^n = \mathbf{T}_h(M_y) \Cs(M_x)\mathbf{x}^n = \mathbf{A}_{(M_y, M_x)}\mathbf{x}^n,
\label{eq:z=Ax}
\end{align}
 where 
 \begin{align}
 	\mathbf{A}_{(M_y, M_x)} = \mathbf{T}_h(M_y) \Cs(M_x) \in \mathbb{C}^{K \times M_x},
 	\label{eq:Amat}
 \end{align}
is equivalent to the compression matrix in standard sparse recovery. The subscripts of $\mathbf{A}$ denote the number of columns of matrices $\mathbf{T}_h$ and $\Cs$. In the current setup, $\mathbf{A}$ is unknown, unlike the standard sparse recovery problem where the matrix is known. For identifiability, the spark property of the matrix $\mathbf{A}$ plays a crucial role, as summarized in the following theorem. 
\begin{theorem}[Sparse-Filters Identifiability for Compressive MBD]
	\label{theorem1}
	{Consider the MBD measurements $\mathbf{y}^n = \mathbf{s}_0*\mathbf{x}^n, n = 1, \ldots, N$ where $N\geq 2$. The source $\mathbf{s}_0$ satisfies assumption (A2) and the sparse filters $\{\mathbf{x}^n\}_{n=1}^N$ satisfy assumption (A1) with sparsity level $L< \min\left(\sqrt{M_{s_0}/2}, \sqrt{M_x} \right)$. Consider the measurements  $\mathbf{z}^n[m], m \in \mathcal{K}$ or $\bar{\mathbf{z}}^n = \mathbf{A}_{(M_y, M_x)} \mathbf{x}^n, n = 1, \ldots, N$. The measurements are compressed, and we have the following necessary and sufficient conditions for identifiability. 		\begin{enumerate}
			\item (Necessary) If $\text{Spark}(\mathbf{A}_{(M_y, M_x)}) \leq 2L$ then the sparse filters are not uniquely identifiable.
			\item (Sufficient) If $\text{Spark}(\mathbf{A}_{(M_y+M_x-1, 2M_x-1)}) > 2L^2$ and either (A3) or (A3$^{\prime}$) is satisfied, then the sparse filters can be identified up to a scaling constant.
		\end{enumerate}
	}
\end{theorem}
\begin{proof}
	We first focus on the necessary and sufficient conditions and then show that for both cases the measurements are compressed. \\
	\noindent\emph{Necessary Conditions:} 
	From \eqref{eq:z=Ax}, we know that if $\text{Spark}(\mathbf{A}_{(M_y, M_x)}) \leq 2L$ then for every sparse filter $\mathbf{x}^n$ there exists another $L$-sparse vector $\mathbf{\hat{x}}^n \neq \mathbf{x}^n$ such that, 
	\begin{align}
		\mathbf{A}_{(M_y, M_x)} (\mathbf{x}^n- \mathbf{\hat{x}}^n) = \mathbf{0},\quad n=1, \ldots, N. 
		\label{eq:necessary1}
	\end{align}
	Hence, there exists an alternative set of sparse filters $\{\mathbf{\hat{x}}^n\}_{n=1}^N$ satisfying the compressed measurements.
	
	\noindent\emph{Sufficient Conditions:} 
	Let us assume that there exists an alternative source $\mathbf{\hat{s}}_0$ and a set of feasible sparse filters $\{\mathbf{\hat{x}}\}_{n=1}^N$ such that 
	\begin{align}
		\bar{\mathbf{z}}^n = \mathbf{h} \underset{t}{\ast} (\mathbf{s}_0*\mathbf{x}^n) = \mathbf{h} \underset{t}{\ast} (\mathbf{\hat{s}}_0*\mathbf{\hat{x}}^n).
		\label{eq:subsampled1}
	\end{align}
	This implies that the alternative set $\left(\mathbf{\hat{s}}_0, \{\mathbf{\hat{x}}\}_{n=1}^N \right)$ results in the same compressed measurements. Since all the vectors are of finite dimension, from \eqref{eq:subsampled1} it can be shown that
	\begin{align}
		 \mathbf{h}*\mathbf{s}_0*\mathbf{x}^n =\mathbf{h}*\mathbf{\hat{s}}_0*\mathbf{\hat{x}}^n, \quad n  = 1, \ldots, N.
		\label{eq:full_mes_alt}
	\end{align}
	 Then, for any two channels, say from channel numbers $n= 1, 2$, we have the following cross-convolution relationship \cite{mulleti_mbd}:
	 	\begin{align}
	 	\mathbf{h} \underset{t}{\ast} (\mathbf{s}_0 * \underbrace{(\mathbf{\hat{x}}_1 {\ast} \mathbf{x}_2 - \mathbf{\hat{x}}_2 {\ast} \mathbf{x}_1)}_{\mathbf{q}})
	 	= \mathbf{0},
	 	\label{eq:s*q=0}
	 \end{align}
	 where $\mathbf{q} \in \mathbb{R}^{2M_x-1}$. From the condition $L<\sqrt{M_x}$, we have that $\|\mathbf{q}\|_0 \leq 2L^2$ \cite{mulleti_mbd}. Using \eqref{eq:Amat}, we rewrite \eqref{eq:s*q=0} as 
	 \begin{align}
	 	\mathbf{A}_{(M_y+M_x-1, 2M_x-1)}\,\mathbf{q} = \mathbf{0}.
	 	\label{eq:Aq=0}
	 \end{align}
	Since $\|\mathbf{q}\|_0 \leq 2L^2$, from \eqref{eq:Aq=0} we infer that $\mathbf{q} = \mathbf{0}$ if $\text{Spark}(\mathbf{A}_{(M_y+M_x-1, 2M_x-1)}) \geq 2L^2$. Then, for $\mathbf{q}  = \mathbf{0}$, we have that $\mathbf{\hat{x}}_1 {\ast} \mathbf{x}_2 = \mathbf{\hat{x}}_2 {\ast} \mathbf{x}_1$ which can be written in the $z$-domain as
	\begin{align}
		\hat{X}_1(z) X_2(z) = \hat{X}_2(z) X_1(z).
		\label{eq:zdomain}
	\end{align}
	This relationship implies that there exists a function $G(z)$ such that $X^n(z) = G(z) \hat{X}^n(z)$. Hence, it can be shown that $G(z)$ is a constant function if one of the conditions (A3) or (A3$^\prime$) is satisfied (See \cite[Section IV.E]{mulleti_mbd}). In other words, $\mathbf{\hat{x}}^n$ is a scaled version of $\mathbf{x}^n$ and hence a feasible solution. A similar argument can be used for any pair of sparse filters, and we conclude that identifiability holds for all sparse filters. 
	
	\noindent \emph{Compressed Measurements:} 
	Next, we show that the number of measurements $K<M_y$, i.e., the measurements $\mathbf{z}^n$ are compressed. The identifiability results are derived in terms of spark properties of $\mathbf{A}$. For the spark properties to be satisfied, the length of the compression filter $M_h$ requires satisfying certain conditions, as discussed next. In the discussion, we use the fact that for any matrix to satisfy the minimum spark property, its number of rows should be greater than or equal to the minimum spark. For example, consider a matrix $\mathbf{B} \in \mathbb{C}^{K \times N}$. Then, for any integer $L\leq N$, the condition $\text{Spark}(\mathbf{B})>L$ is satisfied only if $K\geq L$. 
	
	In the sufficiency part, the inequality $\text{Spark}(\mathbf{A}_{(M_y, M_x)}) \leq 2L$ indicates a necessary condition. Hence, for identifiability it is necessary that $\text{Spark}(\mathbf{A}_{(M_y, M_x)})> 2L$ which is ensured only if the number of rows of $\mathbf{A}_{(M_y, M_x)}$ or more specifically, the number of rows of $\mathbf{T}_h(M_y)$ is greater than or equal to $2L$. Since the number of rows of $\mathbf{A}_{(M_y, M_x)}$ determines the number of compressed measurements (see \eqref{eq:z=Ax}), the results imply that identifiability is not possible if the number of measurements $K$, which is equal to the number of rows of $\mathbf{A}_{(M_y, M_x)}$, is less than $2L$. For $K\geq 2L$, the compression filter's length $M_h$ should be chosen such that $M_h \geq 2L+M_y-1$. Since $L< \min\left(\sqrt{M_{s_0}/2}, \sqrt{M_x} \right)$, we have that 
	\begin{align}
	    M_{s_0}>2L^2 \quad \text{and} \quad M_x>L^2.
	    \label{eq:source_filt_sparsiy}
	\end{align}
	By using the above bounds, it can be verified that $2L\leq K<M_y = M_x+M_{s_0}-1 < 3L^2-1$ and hence the measurements are compressed.
	
	For the sufficient conditions, it is required that $\text{Spark}(\mathbf{A}_{(M_y+M_x-1, 2M_x-1)}) > 2L^2$ which in turn means that the number of rows of $\mathbf{T}_h(M_y+M_x-1)$ is greater than or equal to $2L^2$. This can be insured by choosing the length of the filter such that $M_h\geq 2L^2+2M_x+M_{s_0}-3$. From the spark condition in the sufficient part, it may not be straightforward to determine the number of measurements, as the condition is derived for $\mathbf{A}_{(M_y+M_x-1, 2M_x-1)}$ which is different from $\mathbf{A}_{(M_y, M_x)}$ used for compression (see \eqref{eq:z=Ax}). However, by using the condition $M_h\geq 2L^2+2M_x+M_{s_0}-3$ and calculating the number of rows of $\mathbf{A}_{(M_y, M_x)}$, we note that a minimum of $K = M_h-M_y+1 = 2L^2+M_x-1$  measurements is available. For $K = 2L^2+M_x-1$ and $M_y = M_x+M_{s_0}-1$, it can be verified that $K<M_y$ as $2L^2 < M_{s_0}$. Hence, the measurements are compressed.
\end{proof}

The theorem states necessary and sufficient conditions on the source and compression filter in terms of spark properties. In addition, we show that the measurements are compressible if the sparsity level is restricted in terms of the lengths of the source and the sparse filter. In the frequency-domain guarantees derived in \cite{mulleti_mbd}, the sparsity level should satisfy the inequality $L < \sqrt{M_x}$ for identifiability. However, in the time-domain guarantees (\Cref{theorem1}), it should satisfy the condition $L < \min(\sqrt{M_s}, \sqrt{M_x})$ which is a stronger condition if $M_s<M_x$. For $M_s<M_x$, the time-domain results can identify less sparse signals. This is a consequence of the joint identifiability condition on both the compression filter and the MBD measurements. 

Compared to the identifiability condition in \cite{mulleti_mbd}, the results in~\Cref{theorem1} are generic. To be specific, in the frequency-domain identifiability results \cite{mulleti_mbd}, a D-SoS filter is used together with a non-vanishing spectral support condition on the source. Whereas the time-domain results in~\Cref{theorem1} are generic and do not use any structure or condition on the compression filter or source. Hence, the conditions derived can be used to design a wider class of signals for compression. Note that the D-SoS filter, together with the non-vanishing spectral support condition, also satisfies the spark conditions in~\Cref{theorem1} (See Appendix for details). For sufficient guarantees, the approach in \cite{mulleti_mbd} requires minimum $2L^2$ Fourier measurements, whereas the proposed time-domain approach requires a minimum of $2L^2+M_x-1$ compressed measurements. The additional measurements in the proposed approach are due to the joint identifiability condition. Theoretical guarantees were also presented in \cite{balzano_nowak, lee_17, wang_chi, cosse_mbd} for the S-MBD problem where they consider random sparsity models used. Specifically, these models assume that the sparse filters follow a Bernoulli-Gaussian distribution. The models may not be applicable in many applications. In comparison, the proposed approach is deterministic (no randomness) in nature and hence more widely applicable.
		

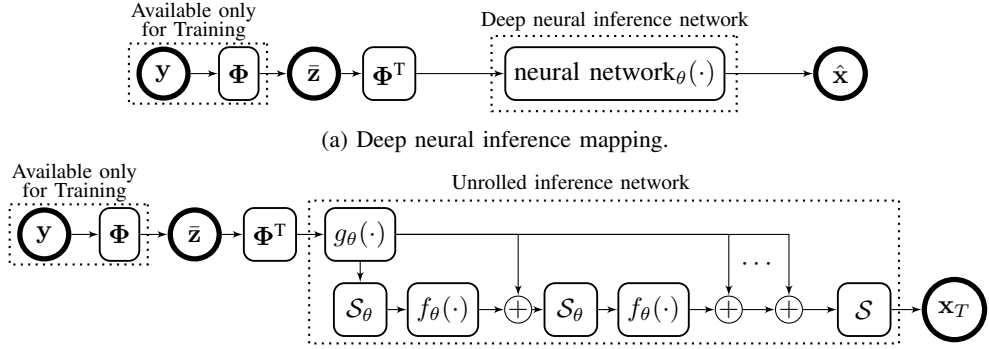
\begin{figure*}[!t]
\begin{subfigure}[t]{0.99\linewidth}
	\begin{minipage}[b]{1.0\linewidth}
		\centering
		\tikzstyle{block} = [draw, fill=none, rectangle, 
		minimum height=2em, minimum width=2em, rounded corners, line width=0.25mm]
		\tikzstyle{sum} = [draw, fill=none, minimum height=0.01em, minimum width=0.01em, inner sep=0pt, circle, node distance=1cm]
		\tikzstyle{cir} = [draw, fill=none, circle, line width=0.7mm, minimum width=0.65cm, node distance=1cm]
		\tikzstyle{loss} = [draw, fill=none, color=black, ellipse, line width=0.5mm, minimum width=0.7cm, node distance=1cm]
		\tikzstyle{input} = [coordinate]
		\tikzstyle{output} = [coordinate]
		\tikzstyle{pinstyle} = [pin edge={to-,thin,black}]
		\begin{tikzpicture}[auto, node distance=2cm,>=latex']
		cloud/.style={
			draw=red,
			thick,
			ellipse,
			fill=none,
			minimum height=1em}
		\node [input, name=input] {};
		\node [cir, node distance=0cm, right of=input] (y) {$\y$};
		\node [block, right of=y,  minimum width=0.5cm, node distance=1.cm] (PHI) {$\mathbf{\Phi}$};
		
		\node [cir, node distance=1.cm, right of=PHI] (z) {$\bar \z$};

		\draw[thick,dotted]     ($(PHI.east)+(+0.1,+0.4)$) rectangle ($(y.west)+(-0.1,-0.4)$);
		
		\node [block, right of=z,  minimum width=0.5cm, node distance=1cm] (projT) {$\mathbf{\proj^{\text{T}}}$};
        
 		\node [block, right of=projT,  minimum width=0.5cm, node distance=3cm] (nn) {$\text{neural network}_{\theta}(\cdot)$};
 
 		\node [cir, right of=nn, node distance=3cm] (x) {$\hat \x$};
 		
		\node [rectangle, fill=none,  node distance=0.82cm,  left=17.0pt,above of=PHI] (training) {\footnotesize{Available only}}; 
		\node [rectangle, fill=none,  node distance=0.57cm,  left=17.0pt,above of=PHI] (training) {\footnotesize{for Training}};
		
		\draw[thick,dotted]     ($(nn.east)+(+0.15,+0.5)$) rectangle ($(nn.west)+(-0.17,-0.5)$);
		\node [rectangle, fill=none,  node distance=0.7cm, above of=nn] (net) {\footnotesize{Deep neural inference network}};

		\draw [->] (y) -- node [] {} (PHI);
		\draw [->] (PHI) -- node [] {} (z);	
		\draw [->] (z) -- node [] {} (projT);
    	\draw [->] (projT) -- node [] {} (nn);

		\draw [->] (nn) -- node[] {} (x);

		\end{tikzpicture}
	\end{minipage}
	\caption{Deep neural inference mapping.}\label{fig:infmap_nn}
\end{subfigure}
\begin{subfigure}[t]{0.99\linewidth}
	\begin{minipage}[b]{1.0\linewidth}
		\centering
		\tikzstyle{block} = [draw, fill=none, rectangle, 
		minimum height=2em, minimum width=2em, rounded corners, line width=0.25mm]
		\tikzstyle{sum} = [draw, fill=none, minimum height=0.01em, minimum width=0.01em, inner sep=0pt, circle, node distance=1cm]
		\tikzstyle{cir} = [draw, fill=none, circle, line width=0.7mm, minimum width=0.65cm, node distance=1cm]
		\tikzstyle{loss} = [draw, fill=none, color=black, ellipse, line width=0.5mm, minimum width=0.7cm, node distance=1cm]
		\tikzstyle{input} = [coordinate]
		\tikzstyle{output} = [coordinate]
		\tikzstyle{pinstyle} = [pin edge={to-,thin,black}]
		\begin{tikzpicture}[auto, node distance=2cm,>=latex']
		cloud/.style={
			draw=red,
			thick,
			ellipse,
			fill=none,
			minimum height=1em}
		\node [input, name=input] {};
		\node [cir, node distance=0cm, right of=input] (y) {$\y$};
		\node [block, right of=y,  minimum width=0.5cm, node distance=1.cm] (PHI) {$\mathbf{\Phi}$};
		
		\node [cir, node distance=1.cm, right of=PHI] (z) {$\bar \z$};

		\draw[thick,dotted]     ($(PHI.east)+(+0.1,+0.4)$) rectangle ($(y.west)+(-0.1,-0.4)$);
		
		\node [block, right of=z,  minimum width=0.5cm, node distance=1cm] (projT) {$\mathbf{\proj^{\text{T}}}$};
        
 		\node [block, right of=projT,  minimum width=0.5cm, node distance=1.2cm] (we) {$g_{\theta}(\cdot)$};
 		
		\node [block, below of=we,  minimum width=0.7cm, node distance=1cm] (relu1) {$\mathcal{S}_{\theta}$};
		\node [block, right of=relu1,  minimum width=0.7cm, node distance=1.1cm] (wd1) {$f_{\theta}(\cdot)$};
		\node [sum, right of=wd1, node distance=1cm] (sum1) {$+$};

		\node [block, right of=sum1,  minimum width=0.7cm, node distance=0.7cm] (relu2) {$\mathcal{S}_{\theta}$};
		\node [block, right of=relu2,  minimum width=0.7cm, node distance=1.1cm] (wd2) {$f_{\theta}(\cdot)$};
		\node [sum, right of=wd2, node distance=1cm] (sum2) {$+$};

		\node [sum, right of=sum2, node distance=0.8cm] (sumb) {$+$};
		
		\node [block, right of=sumb,  minimum width=0.7cm, node distance=1cm] (relu3) {$\mathcal{S}$};

		\node [cir, right of=relu3, node distance=1.2cm] (xt) {$\x_{T}$};
	
		\node [rectangle, right of=sum2, node distance=0.4cm, above=0.4cm] (dots) {$\cdots$};	

		\node [rectangle, fill=none,  node distance=0.82cm,  left=17.0pt,above of=PHI] (training) {\footnotesize{Available only}}; 
		\node [rectangle, fill=none,  node distance=0.57cm,  left=17.0pt,above of=PHI] (training) {\footnotesize{for Training}};
		
		\draw[thick,dotted]     ($(relu3.east)+(+0.1,+1.45)$) rectangle ($(we.west)+(-0.2,-1.45)$);
		\node [rectangle, fill=none,  node distance=1.7cm, above of=relu2] (net) {\footnotesize{Unrolled inference network}};

		\draw [->] (y) -- node [] {} (PHI);
		\draw [->] (PHI) -- node [] {} (z);	
		\draw [->] (z) -- node [] {} (projT);
    	\draw [->] (projT) -- node [] {} (we);
		\draw [->] (we) -- node[] {} (relu1);
		\draw [->] (relu1) -- node[] {} (wd1);
		\draw [->] (wd1) -- node[] {} (sum1);
		\draw [->] (we) -| node[] {} (sum1);
		
		\draw [->] (sum1) -- node[] {} (relu2);
		\draw [->] (relu2) -- node[] {} (wd2);
		\draw [->] (wd2) -- node[] {} (sum2);
		\draw [->] (we) -| node[] {} (sum2);

		\draw [->] (sum2) -- node[] {} (sumb);	
		\draw [->] (we) -| node[] {} (sumb);	
		\draw [->] (sumb) -- node[] {} (relu3);
		\draw [->] (relu3) -- node[] {} (xt);

		\end{tikzpicture}
	\end{minipage}
	\caption{Model-based unrolled inference mapping. Here, $\theta$ denotes the network parameters, and $f_{\theta}$ and $g_{\theta}$ are neural blocks with $\mathcal{S}_{\theta}$ representing the nonlinearity.}\label{fig:infmap_unroll}
\end{subfigure}
	\caption{Inference mapping for compressed multichannel blind-deconvolution.}
	\label{fig:infmap}
\end{figure*}

\begin{figure*}[!t]
\begin{subfigure}[t]{0.9\linewidth}
	\begin{minipage}[b]{1.0\linewidth}
		\centering
		\tikzstyle{block} = [draw, fill=none, rectangle, 
		minimum height=2em, minimum width=2em, rounded corners, line width=0.25mm]
		\tikzstyle{sum} = [draw, fill=none, minimum height=0.01em, minimum width=0.01em, inner sep=0em, circle, node distance=1cm]
		\tikzstyle{cir} = [draw, fill=none, circle, line width=0.7mm, minimum width=0.65cm, node distance=1cm]
		\tikzstyle{loss} = [draw, fill=none, color=black, ellipse, line width=0.5mm, minimum width=0.7cm, node distance=1cm]
		\tikzstyle{input} = [coordinate]
		\tikzstyle{output} = [coordinate]
		\tikzstyle{pinstyle} = [pin edge={to-,thin,black}]
		\begin{tikzpicture}[auto, node distance=2cm,>=latex']
		cloud/.style={
			draw=red,
			thick,
			ellipse,
			fill=none,
			minimum height=1em}
		\node [input, name=input] {};
		\node [cir, node distance=0cm, right of=input] (y) {$\y$};
		\node [block, right of=y,  minimum width=0.5cm, node distance=1.cm] (PHI) {$\mathbf{\Phi}$};
		
		\node [cir, node distance=1.cm, right of=PHI] (z) {$\bar \z$};

		\draw[thick,dotted]     ($(PHI.east)+(+0.1,+0.4)$) rectangle ($(y.west)+(-0.1,-0.4)$);
		
		\node [block, right of=z,  minimum width=0.5cm, node distance=1cm] (projT) {$\mathbf{\proj^{\text{T}}}$};
        
 		\node [block, right of=projT,  minimum width=0.5cm, node distance=1.2cm] (we) {$\alpha\Cs^{\text{T}}$};
 		
		\node [block, below of=we,  minimum width=0.7cm, node distance=1cm] (relu1) {$\mathcal{S}_{\alpha \lambda}$};
		\node [block, right of=relu1,  minimum width=0.7cm, node distance=2.0cm] (wd1) {$\mathbf{I} - \alpha \Cs^{\text{T}}\proj^{\text{T}} \proj \Cs$};
		\node [sum, right of=wd1, node distance=1.8cm] (sum1) {$+$};

		\node [block, right of=sum1,  minimum width=0.7cm, node distance=0.8cm] (relu2) {$\mathcal{S}_{\alpha \lambda}$};
		\node [block, right of=relu2,  minimum width=0.7cm, node distance=2.0cm] (wd2) {$\mathbf{I} - \alpha \Cs^{\text{T}}\proj^{\text{T}} \proj \Cs$};
		\node [sum, right of=wd2, node distance=1.8cm] (sum2) {$+$};

		\node [sum, right of=sum2, node distance=0.8cm] (sumb) {$+$};
		
		\node [block, right of=sumb,  minimum width=0.7cm, node distance=1cm] (relu3) {$\mathcal{S}_{\alpha \lambda}$};

		\node [cir, right of=relu3, node distance=1.2cm] (xt) {$\x_{T}$};
	
		\node [rectangle, right of=sum2, node distance=0.4cm, above=0.4cm] (dots) {$\cdots$};	

		\node [rectangle, fill=none,  node distance=0.82cm,  left=17.0pt,above of=PHI] (training) {\footnotesize{Available only}}; 
		\node [rectangle, fill=none,  node distance=0.57cm,  left=17.0pt,above of=PHI] (training) {\footnotesize{for Training}};
		
		\draw[thick,dotted]     ($(relu3.east)+(+0.1,+1.45)$) rectangle ($(we.west)+(-0.2,-1.45)$);
		\node [rectangle, fill=none,  node distance=1.7cm, above of=relu2] (net) {\footnotesize{Unrolled inference network}};

		\draw [->] (y) -- node [] {} (PHI);
		\draw [->] (PHI) -- node [] {} (z);	
		\draw [->] (z) -- node [] {} (projT);
    	\draw [->] (projT) -- node [] {} (we);
		\draw [->] (we) -- node[] {} (relu1);
		\draw [->] (relu1) -- node[] {} (wd1);
		\draw [->] (wd1) -- node[] {} (sum1);
		\draw [->] (we) -| node[] {} (sum1);
		
		\draw [->] (sum1) -- node[] {} (relu2);
		\draw [->] (relu2) -- node[] {} (wd2);
		\draw [->] (wd2) -- node[] {} (sum2);
		\draw [->] (we) -| node[] {} (sum2);

		\draw [->] (sum2) -- node[] {} (sumb);	
		\draw [->] (we) -| node[] {} (sumb);	
		\draw [->] (sumb) -- node[] {} (relu3);
		\draw [->] (relu3) -- node[] {} (xt);

		\end{tikzpicture}
	\end{minipage}
	\caption{LS-MBD using convolutional ISTA (LS-MBD-CISTA).}\label{fig:ls-mbd-ista}
\end{subfigure}
\begin{subfigure}[t]{0.9\linewidth}
	\begin{minipage}[b]{1.0\linewidth}
		\centering
		\centering
		\tikzstyle{block} = [draw, fill=none, rectangle, 
		minimum height=2em, minimum width=2em, rounded corners, line width=0.25mm]
		\tikzstyle{sum} = [draw, fill=none, minimum height=0.01em, minimum width=0.01em, inner sep=0pt, circle, node distance=1cm]
		\tikzstyle{cir} = [draw, fill=none, circle, line width=0.7mm, minimum width=0.65cm, node distance=1cm]
		\tikzstyle{loss} = [draw, fill=none, color=black, ellipse, line width=0.5mm, minimum width=0.7cm, node distance=1cm]
		\tikzstyle{input} = [coordinate]
		\tikzstyle{output} = [coordinate]
		\tikzstyle{pinstyle} = [pin edge={to-,thin,black}]
		\begin{tikzpicture}[auto, node distance=2cm,>=latex']
		cloud/.style={
			draw=red,
			thick,
			ellipse,
			fill=none,
			minimum height=1em}
		\node [input, name=input] {};
		\node [cir, node distance=0cm, right of=input] (y) {$\y$};
		\node [block, right of=y,  minimum width=0.5cm, node distance=1.cm] (PHI) {$\mathbf{\Phi}$};
		
		\node [cir, node distance=1.cm, right of=PHI] (z) {$\bar \z$};

		\draw[thick,dotted]     ($(PHI.east)+(+0.1,+0.4)$) rectangle ($(y.west)+(-0.1,-0.4)$);
		
		\node [block, right of=z,  minimum width=0.5cm, node distance=1cm] (projT) {$\mathbf{\proj^{\text{T}}}$};
        
 		\node [block, right of=projT,  minimum width=0.5cm, node distance=1.2cm] (we) {$\W_e$};
 		
		\node [block, below of=we,  minimum width=0.7cm, node distance=1cm] (relu1) {$\mathcal{S}_{\bvec}$};
		\node [block, right of=relu1,  minimum width=0.7cm, node distance=2.0cm] (wd1) {$\mathbf{I} - \W_c\proj^{\text{T}} \proj \W_d$};
		\node [sum, right of=wd1, node distance=1.8cm] (sum1) {$+$};

		\node [block, right of=sum1,  minimum width=0.7cm, node distance=0.8cm] (relu2) {$\mathcal{S}_{\bvec}$};
		\node [block, right of=relu2,  minimum width=0.7cm, node distance=2.0cm] (wd2) {$\mathbf{I} - \W_c\proj^{\text{T}} \proj \W_d$};
		\node [sum, right of=wd2, node distance=1.8cm] (sum2) {$+$};

		\node [sum, right of=sum2, node distance=0.8cm] (sumb) {$+$};
		
		\node [block, right of=sumb,  minimum width=0.7cm, node distance=1cm] (relu3) {$\mathcal{S}_{\bvec}$};

		\node [cir, right of=relu3, node distance=1.2cm] (xt) {$\x_{T}$};
	
		\node [rectangle, right of=sum2, node distance=0.4cm, above=0.4cm] (dots) {$\cdots$};	

		\draw[thick,dotted]     ($(relu3.east)+(+0.1,+1.45)$) rectangle ($(we.west)+(-0.2,-1.45)$);
		\node [rectangle, fill=none,  node distance=1.7cm, above of=relu2] (net) {\footnotesize{Unrolled inference network}};
		
		\draw [->] (y) -- node [] {} (PHI);
		\draw [->] (PHI) -- node [] {} (z);	
		\draw [->] (z) -- node [] {} (projT);	
		\draw [->] (projT) -- node [] {} (we);
		\draw [->] (we) -- node[] {} (relu1);
		\draw [->] (relu1) -- node[] {} (wd1);
		\draw [->] (wd1) -- node[] {} (sum1);
		\draw [->] (we) -| node[] {} (sum1);
		
		\draw [->] (sum1) -- node[] {} (relu2);
		\draw [->] (relu2) -- node[] {} (wd2);
		\draw [->] (wd2) -- node[] {} (sum2);
		\draw [->] (we) -| node[] {} (sum2);

		\draw [->] (sum2) -- node[] {} (sumb);	
		\draw [->] (we) -| node[] {} (sumb);	
		\draw [->] (sumb) -- node[] {} (relu3);
		\draw [->] (relu3) -- node[] {} (xt);

		\end{tikzpicture}
	\end{minipage}
	\caption{LS-MBD using learned convolutional ISTA (LS-MBD-LCISTA).}\label{fig:ls-mbd-lista}
\end{subfigure}
\begin{subfigure}[t]{0.9\linewidth}
	\begin{minipage}[b]{1.0\linewidth}
		\centering
		\tikzstyle{block} = [draw, fill=none, rectangle, 
		minimum height=2em, minimum width=2em, rounded corners, line width=0.25mm]
		\tikzstyle{sum} = [draw, fill=none, minimum height=0.01em, minimum width=0.01em, inner sep=0pt, circle, node distance=1cm]
		\tikzstyle{cir} = [draw, fill=none, circle, line width=0.7mm, minimum width=0.65cm, node distance=1cm]
		\tikzstyle{loss} = [draw, fill=none, color=black, ellipse, line width=0.5mm, minimum width=0.7cm, node distance=1cm]
		\tikzstyle{input} = [coordinate]
		\tikzstyle{output} = [coordinate]
		\tikzstyle{pinstyle} = [pin edge={to-,thin,black}]
		\begin{tikzpicture}[auto, node distance=2cm,>=latex']
		cloud/.style={
			draw=red,
			thick,
			ellipse,
			fill=none,
			minimum height=1em}
		\node [input, name=input] {};
		\node [cir, node distance=0cm, right of=input] (y) {$\y$};
		\node [block, right of=y,  minimum width=0.5cm, node distance=1.cm] (PHI) {$\mathbf{\Phi}$};
		
		\node [cir, node distance=1.cm, right of=PHI] (z) {$\bar \z$};

		\draw[thick,dotted]     ($(PHI.east)+(+0.1,+0.4)$) rectangle ($(y.west)+(-0.1,-0.4)$);
		
		\node [block, right of=z,  minimum width=0.5cm, node distance=0.9cm] (projT) {$\mathbf{\proj^{\text{T}}}$};
        
 		\node [block, right of=projT,  minimum width=0.5cm, node distance=1.35cm] (denoiser) {$\text{Denoiser}$};
 		
		\node [block, right of=denoiser,  minimum width=0.5cm, node distance=1.5cm] (we) {$\W_e$}; 
		
		\node [block, below of=we,  minimum width=0.7cm, node distance=1cm] (relu1) {$\mathcal{S}_{\bvec_1}$};
		\node [block, right of=relu1,  minimum width=0.7cm, node distance=2.0cm] (wd1) {$\mathbf{I} - \W_c^1\proj^{\text{T}} \proj \W_d^1$};
		\node [sum, right of=wd1, node distance=1.8cm] (sum1) {$+$};

		\node [block, right of=sum1,  minimum width=0.7cm, node distance=0.8cm] (relu2) {$\mathcal{S}_{\bvec_2}$};
		\node [block, right of=relu2,  minimum width=0.7cm, node distance=2.0cm] (wd2) {$\mathbf{I} - \W_c^2\proj^{\text{T}} \proj \W_d^2$};
		\node [sum, right of=wd2, node distance=1.8cm] (sum2) {$+$};

		\node [sum, right of=sum2, node distance=0.8cm] (sumb) {$+$};
		
		\node [block, right of=sumb,  minimum width=0.7cm, node distance=1cm] (relu3) {$\mathcal{S}_{\bvec_T}$};

		\node [cir, right of=relu3, node distance=1.2cm] (xt) {$\x_{T}$};
	
		\node [rectangle, right of=sum2, node distance=0.4cm, above=0.3cm] (dots) {$\cdots$};	

		\draw[thick,dotted]     ($(relu3.east)+(+0.1,+1.45)$) rectangle ($(we.west)+(-0.2,-1.45)$);
		\node [rectangle, fill=none,  node distance=1.7cm, above of=relu2] (net) {\footnotesize{Unrolled inference network}};
		
		\draw [->] (y) -- node [] {} (PHI);
		\draw [->] (PHI) -- node [] {} (z);	
		\draw [->] (z) -- node [] {} (projT);	
		\draw [->] (projT) -- node [] {} (denoiser);
		\draw [->] (denoiser) -- node [] {} (we);	
		\draw [->] (we) -- node[] {} (relu1);
		\draw [->] (relu1) -- node[] {} (wd1);
		\draw [->] (wd1) -- node[] {} (sum1);
		\draw [->] (we) -| node[] {} (sum1);
		
		\draw [->] (sum1) -- node[] {} (relu2);
		\draw [->] (relu2) -- node[] {} (wd2);
		\draw [->] (wd2) -- node[] {} (sum2);
		\draw [->] (we) -| node[] {} (sum2);

		\draw [->] (sum2) -- node[] {} (sumb);	
		\draw [->] (we) -| node[] {} (sumb);	
		\draw [->] (sumb) -- node[] {} (relu3);
		\draw [->] (relu3) -- node[] {} (xt);

		\end{tikzpicture}
	\end{minipage}
	\caption{LS-MBD using learned convolutional ISTA with untied layers denoiser (LS-MBD-LCISTA-denoiser). }\label{fig:ls-mbd-lista-free-denoiser}
\end{subfigure}
	\label{fig:ls-mbd}
\end{figure*}

Given the above-mentioned identifiability results for the recovery of sparse filters from compressed measurements, we next focus on formulating the optimization problem and providing algorithms to solve S-MBD.

\section{Unrolled Learning for Compressed Multichannel Blind Deconvolution}
\label{sec:unroll}
Given the generality of the compression filter, we propose to learn the compression filter from data in the time domain. Specifically, we introduce a structured inference network based on algorithm unrolling~\cite{monga2019algorithm} to learn a filter-based compression operator for compressed S-MBD. First, we formulate the problem as an alternating-minimization-based optimization and review inference networks in the context of compressive MBD. Then, we introduce constrained and unrolled networks to structure the inference mapping, on which our main results are based.
\subsection{Learned optimization-based methods}
Given the generative model, we form the following optimization problem to recover the source $\s_0$, estimate the sparse filters $\x^n$, and learn a compression operator $\proj$ (i.e., $\mathbf{T}_h(M_y)$):
\begin{equation}\label{eq:opt}
\begin{aligned}
\min_{\proj, \s_0, \{\x^n\}_{n=1}^N}&\ \sum_{n=1}^N \frac{1}{2} \| \bar \z^n - \proj (\s_0 * \x^n) \|_2^2 + \lambda \| \x^n \|_1\\
&\text{s.t.}\ \|\s_0\|_2 = \| \h \|_2 = 1,
\end{aligned}
\end{equation}
where $\lambda$ is a regularization parameter enforcing sparsity, and $\h$ is the filter corresponding to the compression operator $\proj$. $\proj$ is implemented as a convolution operator followed by a truncation. To avoid scaling ambiguity, we impose a norm constraint on the source $\s_0$ and compression filter $\h$. We can solve this problem by alternating minimization. At update iteration $l$, given an estimate of $\s_0^{(l)}$ and $\proj^{(l)}$, solve the sparse coding problem for each example, i.e.,
\begin{equation}\label{eq:opt_sc}
\begin{aligned}
\x^{n(l)} = \argmin_{\x^n} \frac{1}{2} \| \bar \z^n - \proj^{(l)} (\s_0^{(l)} * \x^n) \|_2^2 + \lambda \| \x^n \|_1.
\end{aligned}
\end{equation}
We call this step sparse coding. Then, use the newly estimated sparse filters $\x^{n(l)}$ to refine the learned $\s_0$ and $\h$, i.e.,
\begin{equation}\label{eq:opt_dict}
\begin{aligned}
\{\proj^{(l+1)}, \s^{(l+1)}\} = \argmin_{\proj, \s_0}&\ \sum_{n=1}^N \frac{1}{2} \| \bar \z^n - \proj (\s_0 * \x^{n(l)})\|_2^2\\
&\text{s.t.}\ \|\s_0\|_2 = \| \h \|_2 = 1
\end{aligned}
\end{equation}
which is the learning step. Next, we review inference networks and discuss how to use the above alternating-minimization-based formulation to construct structured deep neural networks for S-MBD.

\subsection{Inference networks}

Following a fully data-driven approach, one may utilize standard deep learning-based inference networks; this is to learn a mapping from compressed measurements $\bar \z$ to the sparse filters $\x$. For training, the full measurements $\y$ are compressed using the trainable compression operator $\proj$, and then $\proj^{\text{T}}$ is applied to have an approximation of $\y$; this is followed by an inference network to output $\hat \x$. The joint learning of $\proj$ and $\x$ is achieved by forward and backward passes. The forward pass infers the sparse filter $\x$, and $\proj$ is learned in the backward pass during training. This approach is used previously by~\cite{mousavi2017deepcodec, mousavi2017learninverse} to recover data from compressed measurements, where the inference mapping is a deep neural network (\cref{fig:infmap_nn}). In this paper, we propose to give an unrolling structure~\cite{monga2019algorithm} to the inference network (\cref{fig:infmap_unroll}); the building blocks $g_{\theta}$, $f_{\theta}$, and $\mathcal{S}_{\theta}$ are determined by the optimization model, chosen to solve the problem. Prior work has demonstrated the competitive performance of unrolled networks, with an order of magnitude lower number of trainable parameters, to a state-of-the-art deep network on the supervised task of Poisson image denoising~\cite{tolooshams2020icml}. To complement this, we show in~\cref{sec:fewshot} that indeed our unrolled learning framework (\cref{fig:infmap_unroll}) outperforms the generic neural network-based inference mappings (\cref{fig:infmap_nn}) in data-limited regimes.

\subsubsection{Constrained unrolled networks}

To construct the inference mapping, we use an approach similar to~\cite{tolooshams2020tnnls} to map the alternating optimization problem into an autoencoder. The forward pass is designed for sparse filter recovery separable among the examples (i.e., channels), and the backward pass is used to learn parameters shared among the dataset (i.e., source and compression filter). We design the forward pass of the encoder to solve the sparse coding step~\eqref{eq:opt_sc} by performing $T$ iterations of iterative soft thresholding algorithm (ISTA)~\cite{daubehies2004ista}. Each iteration implements
\begin{equation}\label{eq:ista}
\x_{t+1} = \mathcal{S}_{\alpha \lambda}(\x_{t} - \alpha \Cs^{\text{T}}\proj^{\text{T}} (\proj \Cs \x_{t} - \bar \z)),
\end{equation}
where $\alpha$ is the step-size, $\mathcal{S}_b(v) = \text{ReLU}_b(v) - \text{ReLU}_b(-v)$ is the shrinkage operator with $\text{ReLU}_b(v) = (v - b) \cdot \mathds{1}_{v \geq b}$, and $\Cs$ is the convolution matrix corresponding to the source $\s_0$. The decoder maps the code estimates $\x_T^n$ into an estimate of the data $\hat \y^n$ and then the compressed data $\hat{\bar{\z}}^n$ using $\s_0$ and $\h$. The learning step~\eqref{eq:opt_dict} minimizes the loss
\begin{equation}
\begin{aligned}
   &\min_{\proj, \Cs}\ \sum_{n=1}^N \frac{1}{2} \| \bar \z^n - \proj \Cs \x_T^n(\Cs, \proj) \|_2^2,\\
    &\text{s.t.}\ \|\s_0\|_2 = \| \h \|_2 = 1
\end{aligned}
\end{equation}
achieved by backpropagated gradient descent, followed by a projection step to ensure the norm constraints. Having access to the full measurements $\y^n$, the learning step~\eqref{eq:opt_dict} can be modified to $\min_{\proj, \Cs}\ \sum_{n=1}^N \frac{1}{2} \| \y - \Cs \x_T^n(\Cs, \proj) \|_2^2$ or $\min_{\proj, \Cs}\ \sum_{n=1}^N \frac{1}{2} \| \x^n - \x_T^n(\Cs, \proj) \|_2^2$, in the supervised setting, with the knowledge of the sparse codes. Indeed, this is possible as the sparse filters are being estimated through the inference network and are the function of $\Cs$ and $\proj$. Hence, the gradient of the loss $\| \x^n - \x_T^n \|_2^2$ with respect to $\Cs$ and $\proj$ can be computed by backpropagation. We use this sparse filter loss to train the proposed architectures.

\cref{fig:ls-mbd-ista} shows the encoder network architecture. In the context of dense dictionary learning and given full measurements, Tolooshams and Ba \cite{tolooshams2022stable} provide a theoretical analysis of such constrained networks; they provide conditions on the network and data under which dictionary recovery is achieved via network training.

\subsubsection{Learned unrolled networks}
The ISTA encoder architecture can be further modified by relaxing and untying the filter operator corresponding to $\Cs$ and $\Cs^{\text{T}}$ as shown in~\cref{fig:ls-mbd-lista}. We call this architecture learned structured multichannel blind deconvolution learned ISTA (LS-MBD-LISTA). This is to increase the capacity of the model to learn from data. Unrolling ISTA iterations and similar relaxations were introduced first by~\cite{gregor2010learning} for dense matrices, where each iteration of ISTA is referred to as a layer of a structured neural network. Our proposed architecture is a variant of LISTA introduced in~\cite{gregor2010learning} where the relaxed operations $\W_e$, $\W_c$, and $\W_d$, each has their own filters $\h_e, \h_c, \h_d$; $\W_e$ and $\W_c$ performs correlation, and $\W_d$ performs convolution. In this case, we can also let the biases of the neural network $\bvec$ be learned from data.

In the presence of measurement noise, we propose to add a differentiable denoiser (e.g., neural network architecture) before the unrolled layers and specify different $\W_c^t$, and $\W_d^t$ kernels and vector bias $\bvec_t$ for each layer $t$ (\cref{fig:ls-mbd-lista-free-denoiser}). We use the variants of our proposed learned structured multichannel blind deconvolution network in the results discussed in the paper.

\subsection{Training}

Similar to the two-stage approach proposed in~\cite{tolooshams2021lsmbd}, we first train the encoder along with a decoder given full measurements (i.e., $\proj = \mathbf{I}$). This is an unsupervised training for source $\Cs$ recovery by minimizing the reconstruction loss $\| \y - \hat \y \|_2^2$ using the LS-MBD-CISTA network encoder architecture. This network architecture is only used for source recovery purposes. For applications where the source is known, this first stage can be skipped. In the second stage, having access to the recovered $\Cs$ and following the generative model~\eqref{eq:mbd1}, we construct a training set consisting of $\{\y^n, \x^n\}_{n=1}^N$ (See \Cref{sec:data} for more info on data generation). LS-MBD-LCISTA or LS-MBD-LCISTA-denoiser are the main networks used for target recovery in this stage. For a given compression ratio $\text{CR} = \frac{M_h - M_y + 1}{M_y}$, we learn the compression operator $\proj$ and the encoder network parameters in a supervised manner by minimizing $\| \x - \hat \x \|_2^2$. We trained one network for each of the studied scenarios. For example, each CR value has its trained network. Moreover, when studying the effect of source shape, unrolling layers, or measurement noise, one network is trained for each specific scenario.

\section{Results}
\label{sec:results}
This section evaluates the performance of our unrolled learning filter-based compression for multichannel blind deconvolution. We outline the data generation details and explain the training procedure. The section characterizes the proposed method and the effect of unrolling on sparse filter recovery. Moreover, we evaluate unrolled learning against the optimization-based method in~\cite{mulleti_mbd}. Specifically, given the source and sparse targets, we show the outperformance of LS-MBD-LCISTA against FS-MBD in the absence of measurement noise. Then, we offer a more efficient compression operator with minimal reduction in performance. At last, we assess our method, specifically LS-MBD-LCISTA-denoiser, in the presence of noise and highlight that using unrolled learning to structure the inference mapping makes it possible to train networks in data-limited regimes.

 \subsection{Simulated data generation and evaluation criteria}\label{sec:data}
 
 We generated training simulated data with the following parameters ($N=10{,}000, M_{s_0} = 99, M_x = 100, x_s=6$ where the code is $x_s = \|\x\|_0$-sparse.) For most of the experiments, we used a Gaussian source shape, $\s_0[k] = e^{-\frac{6}{5 \sqrt{M_{s_0}}}} (k - \lfloor\frac{M_{s_0}}{2}\rfloor)^2$ where $k \in \{0, 2, \ldots, M_{s_0}-1\}$, followed by a normalization step $\| \s_0 \|_2 = 1$. Gaussian source shapes are mainly concentrated in low frequencies. Hence, to highlight the robustness of our method toward source shapes containing narrow frequency bands, for a portion of the experiments, we switched to a symmetric Morlet wavelet with a scaling factor of $1.5$. The choice of source shape is the same during both training and inference. The sparse filter supports are chosen uniformly at random, with amplitudes drawn from $\text{Unif}(4,9)$. For the noisy experiments, we added Gaussian noise, denoted by $\ve$, with various $\sigma$ and computed the SNR of the signal as $20 \log{\frac{\| \s_0 * \x \|_2}{\| \ve \|_2}}$. Specifically, given the same data distribution across all channels, the reported SNR is averaged over the channel (data examples).
 
 We followed a similar approach and generated two additional test and validation sets, each of size $N=1{,}000$. The validation set is used to choose the best-performing network across the training epochs. All datasets of training, test, and validation follow the same data distribution, i.e., they share similar sparse filter statistics and source shapes.
 
We refer to the sparse filters as code or the target of interest, and evaluate the performance of the methods using three metrics: normalized root-mean-squared error (NRMSE) between the estimate and ground-truth code, exact hit rate $\frac{\text{number of correctly detected targets}}{\text{actual number of targets}}$, and approximate hit rate which we allow a tolerance of one (e.g., for ground-truth code with a support set of $\{7, 45, 90\}$ and estimated code with a support set of $\{7, 46, 89\}$, the approximate hit rate is $100\%$). In this definition, the actual target refers to the support (non-zero entries) of the ground-truth code, and the detected targets refers to the support of the estimated code. We expect a lower NRMSE, and a higher hit rate as the compression ratio CR increases. We report results after picking $x_s = 6$ highest elements of the code support and setting the rest of the code entries to $0$.

\subsection{Training}

The networks are implemented and trained in PyTorch. The convolution weights are randomly initialized using a standard Normal distribution, and the biases, if trained, are set to zero. For LS-MBD-CISTA, we set $\alpha = 0.05$, and for the rest of the networks, $\W_e$ and $\W_c$ are scaled by a factor of $0.05$ after initialization to adjust to the step size. For all networks except LS-MBD-CISTA, the step size $\alpha$ is absorbed, hence learned within the weight $\W_c$. We unrolled all networks for up to $T=30$ layers unless stated. Given the non-negatively of the sparse filters, we set $\mathcal{S} = \text{ReLU}$.

All networks are trained for $10{,}000$ epochs with a batch size of $100$. We use the ADAM optimizer with a learning rate of $10^{-3}$ and $\epsilon = 10^{-14}$~\cite{kingma2014adam}; the learning rate is decayed by a factor of $0.9$ in the middle of the training. The small value of $\epsilon$ is to ensure the exploration of the loss landscape without staying in local minima. The chosen learning rate is the common value used for unrolled networks in the literature~\cite{tolooshams2022stable}.

\subsection{Effect of unrolling}

We characterize the proposed methods as a function of unrolling, i.e., how well the target locations and amplitudes are recovered as the network is unrolled. We used the Gaussian source shape for this section and the architecture type of LS-MBD-LCISTA. We compared unrolling networks of $T = 5$ and $T = 30$ for a range of compression ratios. \cref{fig:unroll} demonstrates the results: the amount of unrolling is crucial for amplitude recovery regardless of compression ratio and exact hit rate, especially when CR is low. The reported results are the average of two independent trials.


\begin{figure}[!t]
\centering
\begin{subfigure}[t]{0.49\linewidth}
\centering
\includegraphics[width=0.999\linewidth]{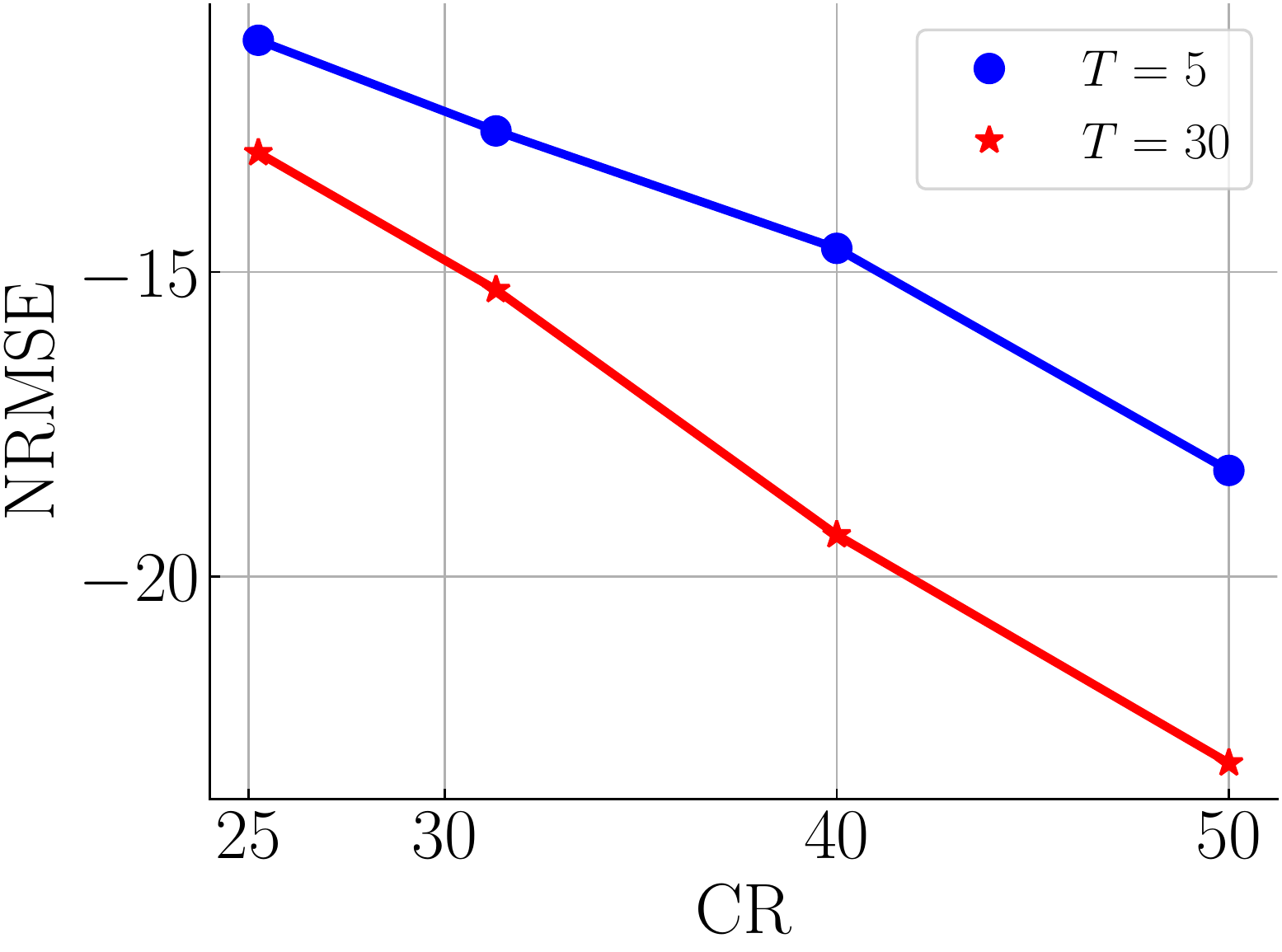}
\caption{Code NRMSE}
\label{fig:unroll_rmse}
\end{subfigure}
\begin{subfigure}[t]{0.49\linewidth}
\centering
\includegraphics[width=0.999\linewidth]{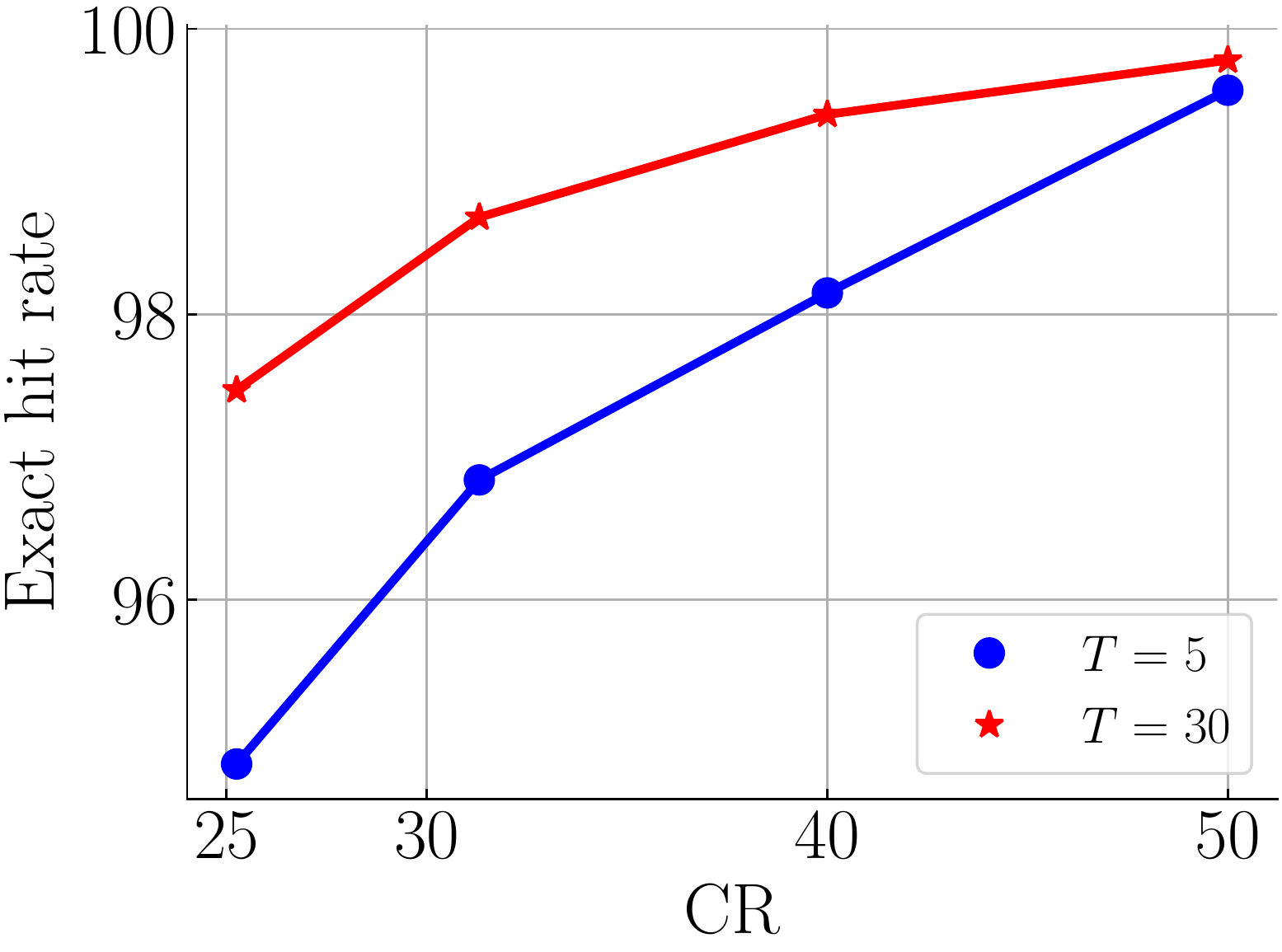}
\caption{Code exact hit rate}
\label{fig:unroll_hitrate0}
\end{subfigure}
\begin{subfigure}[t]{0.49\linewidth}
\centering
\includegraphics[width=0.999\linewidth]{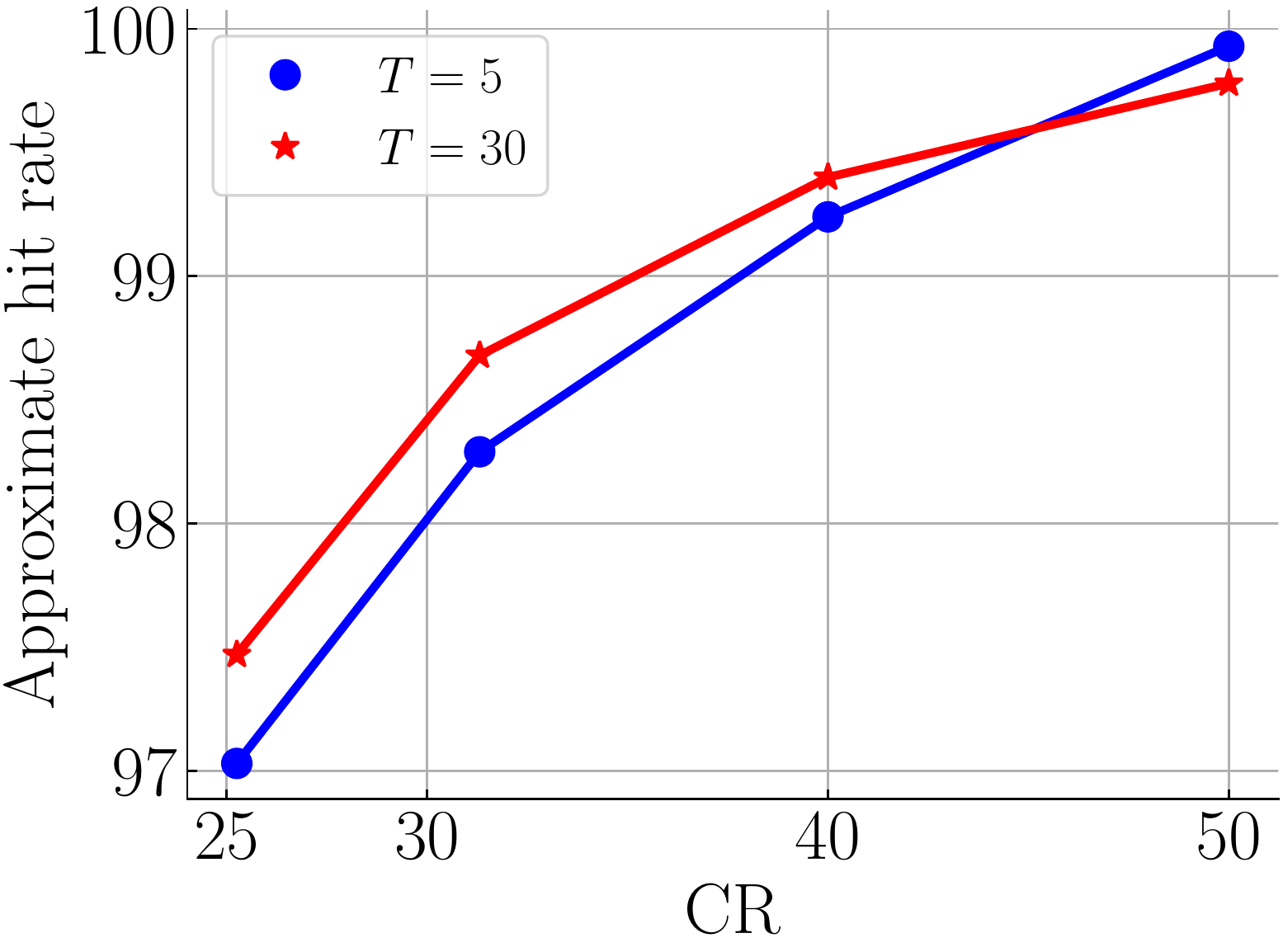}
\caption{Code approximate hit rate}
\label{fig:unroll_hitrate1}
\end{subfigure}
\caption{Effect of Unrolling.}
\label{fig:unroll}
\end{figure}

\subsection{Unrolled learning vs. optimization-based methods}

We compare the proposed network LS-MBD-LCISTA (\cref{fig:ls-mbd-lista}) against the optimization-based method of fixed and structured MBD (FS-MBD)~\cite{mulleti_mbd}; FS-MBD uses source information in the Fourier domain to design a compression matrix $\proj$, uses a blind-dictionary calibration for structured source estimation, and solve the sparse coding step using FISTA~\cite[Ch. 2]{palomar_eldar_2009}\cite{beck2009fast}. Concerning the problem formulation in this manuscript, FS-MBD solves~\eqref{eq:opt} over sparse targets and the source in the frequency domain with known designed compression. For a fair comparison, we assume both of the methods know the source shape; hence, the problem reduces to designing/learning the compression.

We use two different source shapes (Gaussian and symmetric wavelet). For both scenarios and in all compression ratios, 
LS-MBD-LCISTA outperforms FS-MBD. This result conveys two main messages. First, the learnable network parameters, especially the biases, allow LS-MBD-LCISTA to estimate the sparse filters in terms of NRMSE (\cref{fig:fsmbdbaseline_rmse}) much better than FS-MBD which used $\ell_1$-based sparse coding with fixed sparsity level. Second, our method is robust to source shape; the out-performance of LS-MBD-LCISTA against FS-MBD is significant when the source has a wavelet shape (i.e., the frequency content is not concentrated around low frequencies and contains narrow bands). \cref{fig:fsmbdbaseline} visualizes this result in terms of both NRMSE and exact hit rate.

\begin{figure}[!t]
\centering
\begin{subfigure}[t]{0.49\linewidth}
\centering
\includegraphics[width=0.999\linewidth]{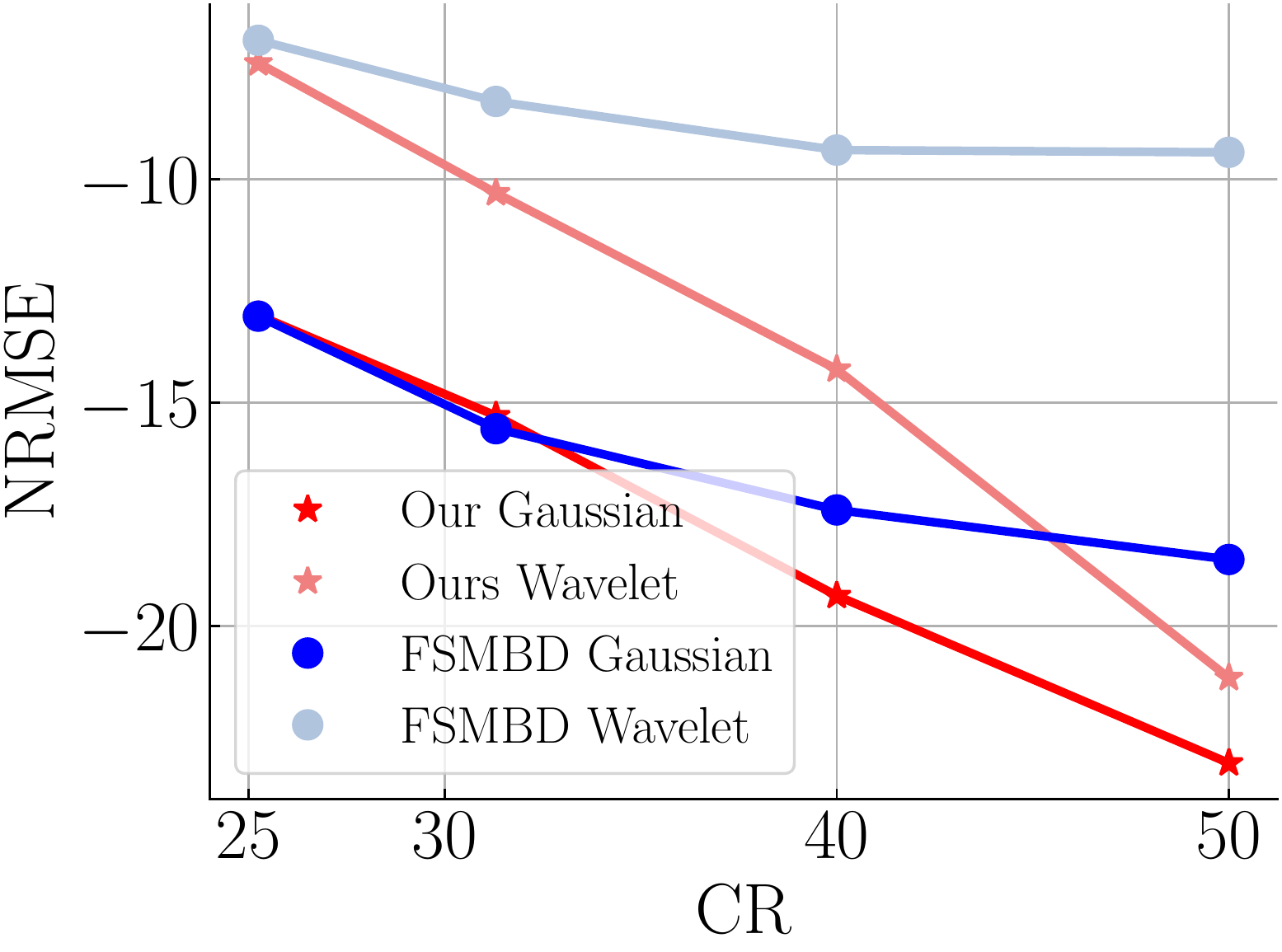}
\caption{Code NRMSE}
\label{fig:fsmbdbaseline_rmse}
\end{subfigure}
\begin{subfigure}[t]{0.49\linewidth}
\centering
\includegraphics[width=0.999\linewidth]{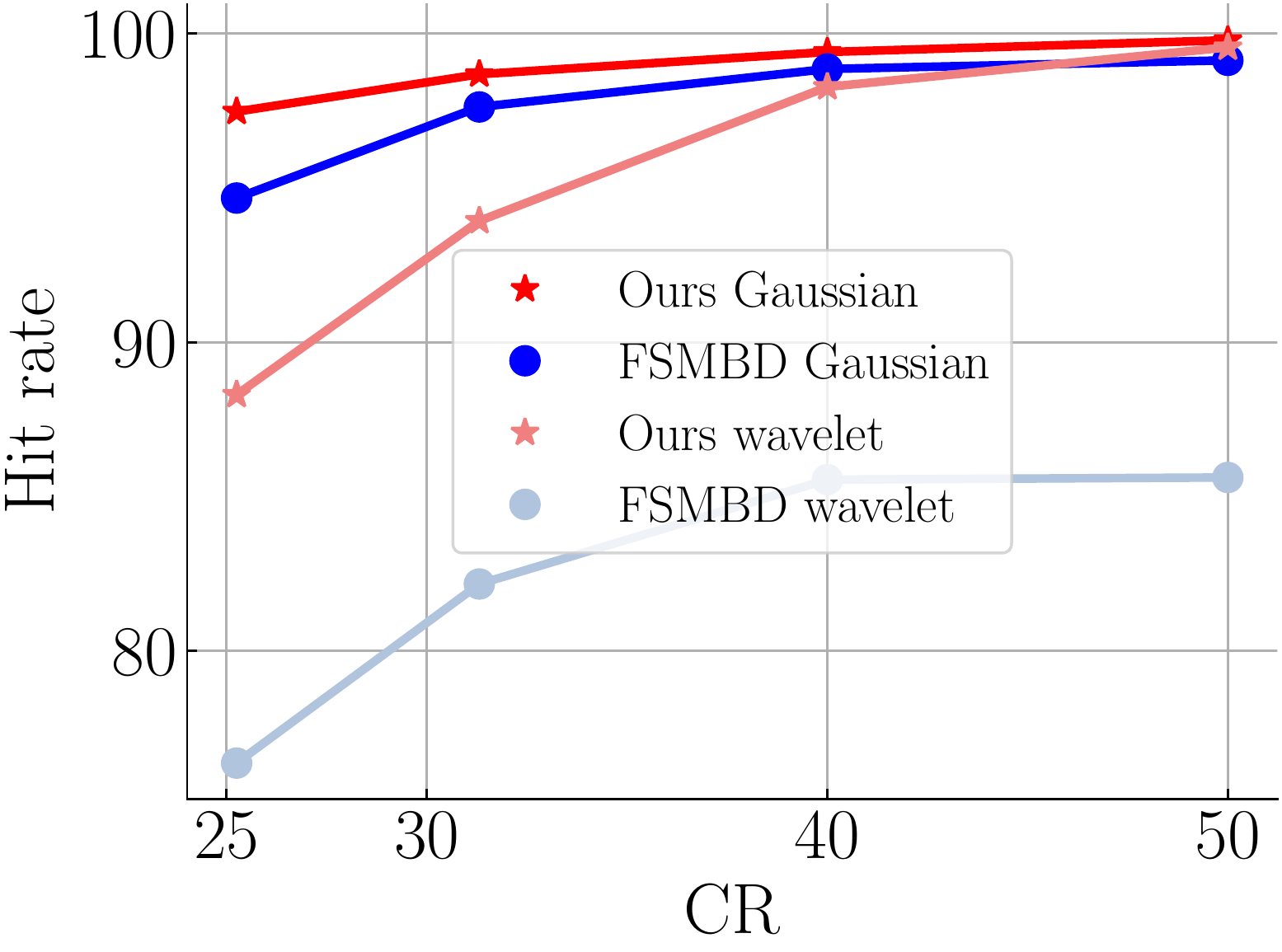}
\caption{Code hit rate}
\label{fig:fsmbdbaseline_hitrate}
\end{subfigure}
\caption{Ours vs. FS-MBD. Our method is more robust to the source shape. FS-MBD only works well when the source is concentrated in low-freq and is not narrow in freq.}
\label{fig:fsmbdbaseline}
\end{figure}

\subsection{Toward more efficient compression operator}

Although the proposed compression operator is efficient compared to unstructured matrix compressions as discussed in~\cite{tolooshams2021lsmbd}, the method has one limitation; The length of the filter $M_h$ must be greater than the length of the signal, and this makes the method not suitable for long data. In this section, we ameliorate this by re-designing the compression operator as follows: the compression performs a sum of $P$ correlations (shown in~\eqref{eq:comp_bank}) where their kernel size $M_h$ is much smaller than the size of the signal. Then, the compression is achieved by decimation (this decimation is also referred to as strides). The desired compression ratio can be achieved by changing the number of convolutional filters and the decimation factor.
\begin{equation}\label{eq:comp_bank}
\bar \z = \proj \y = \sum_{p=1}^P \mathbf{D}_d (\h_p * \y)
\end{equation}
where $\mathbf{D}_d$ denotes the decimation operator matrix by a factor of $d$. Hence, the compression $\proj$ is implemented as a sum of convolutions with strides. We call this method ``filter banks'', and refer to the former one as ``long''. \cref{fig:bank} shows that both methods show similar code recovery performance for various ranges of compression ratios. The data used for this experiment has a symmetric wavelet source shape, and the filter bank compression uses kernels of size $M_h=9$ with stride $9$. For example, for $\text{CR} = \%55.5$, we used $5$ convolutional filters, or for $\text{CR} = \%22.2$, only two filters are used. The architecture used in the above comparison is LS-MBD-LCISTA.

Given the comparable performance, \cref{tab:memory} highlights the memory advantage of the filter banks method compared to the long method. For example, given the generated data (see~\cref{sec:data}) for the long method when $\text{CR} = 50$, the compression length is $M_h = 296$ which is almost twice as large as the signal length. However, for the filter banks methods, compression of $\text{CR} = 55.5$ can be achieved by $5$ filters of each length $M_h = 9$ with decimation (stride) of $9$ (This is indeed independent of the length of the signal). We compute the memory cost based on the dimensionality of the compression operator. This comes down to much smaller memory storage of $5 \times 9 = 45$. For this scenario, \Cref{fig:data_vis} visualizes a test example of full measurements (a), the ground-truth and estimated sparse filters code (b), the source shape (c), the compressed data (d), and the learned compression filter banks (e).

\begin{figure}[!t]
\centering
\begin{subfigure}[t]{0.49\linewidth}
\centering
\includegraphics[width=0.999\linewidth]{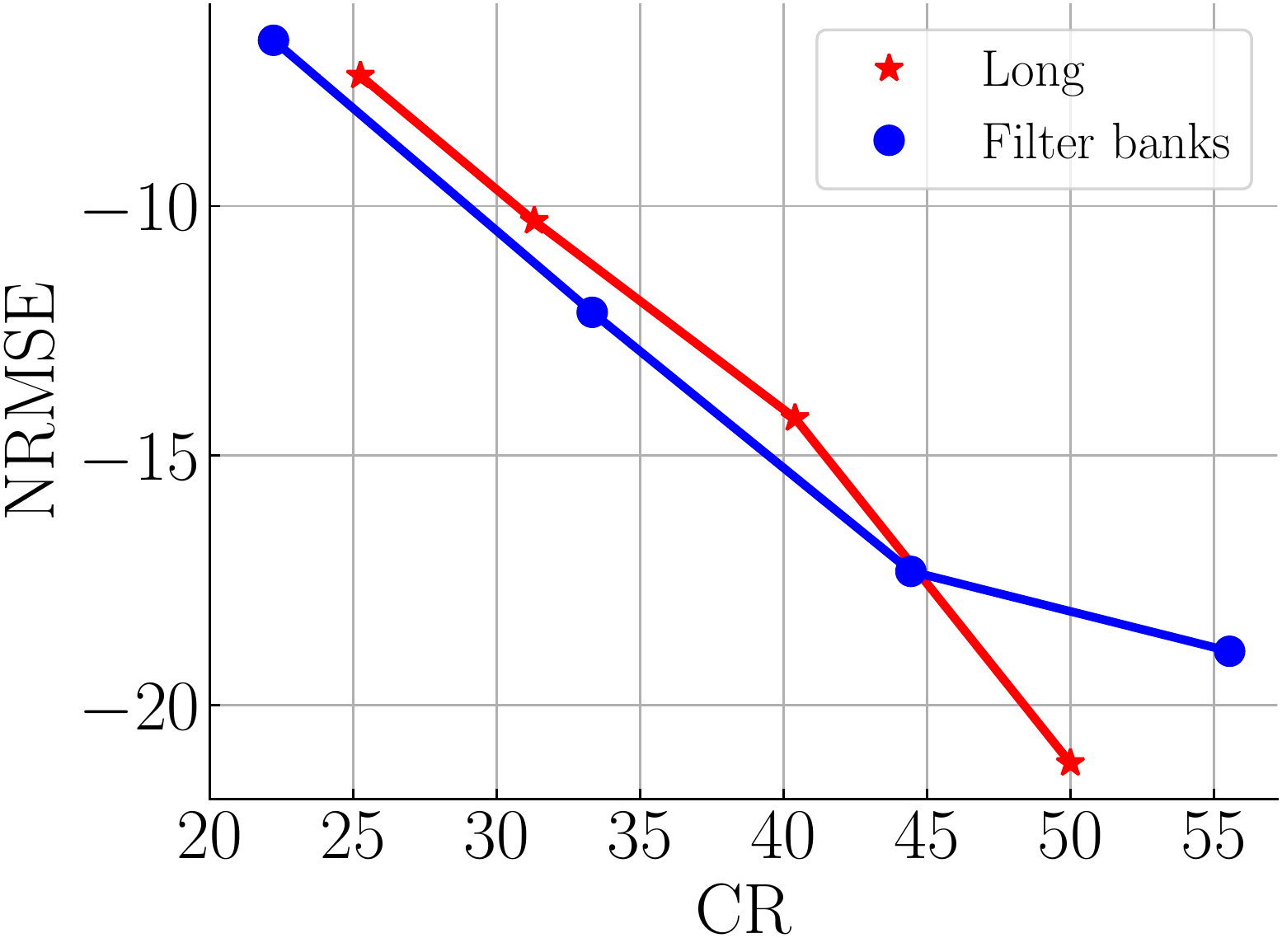}
\caption{Code NRMSE}
\label{fig:bank_rmse}
\end{subfigure}
\begin{subfigure}[t]{0.49\linewidth}
\centering
\includegraphics[width=0.999\linewidth]{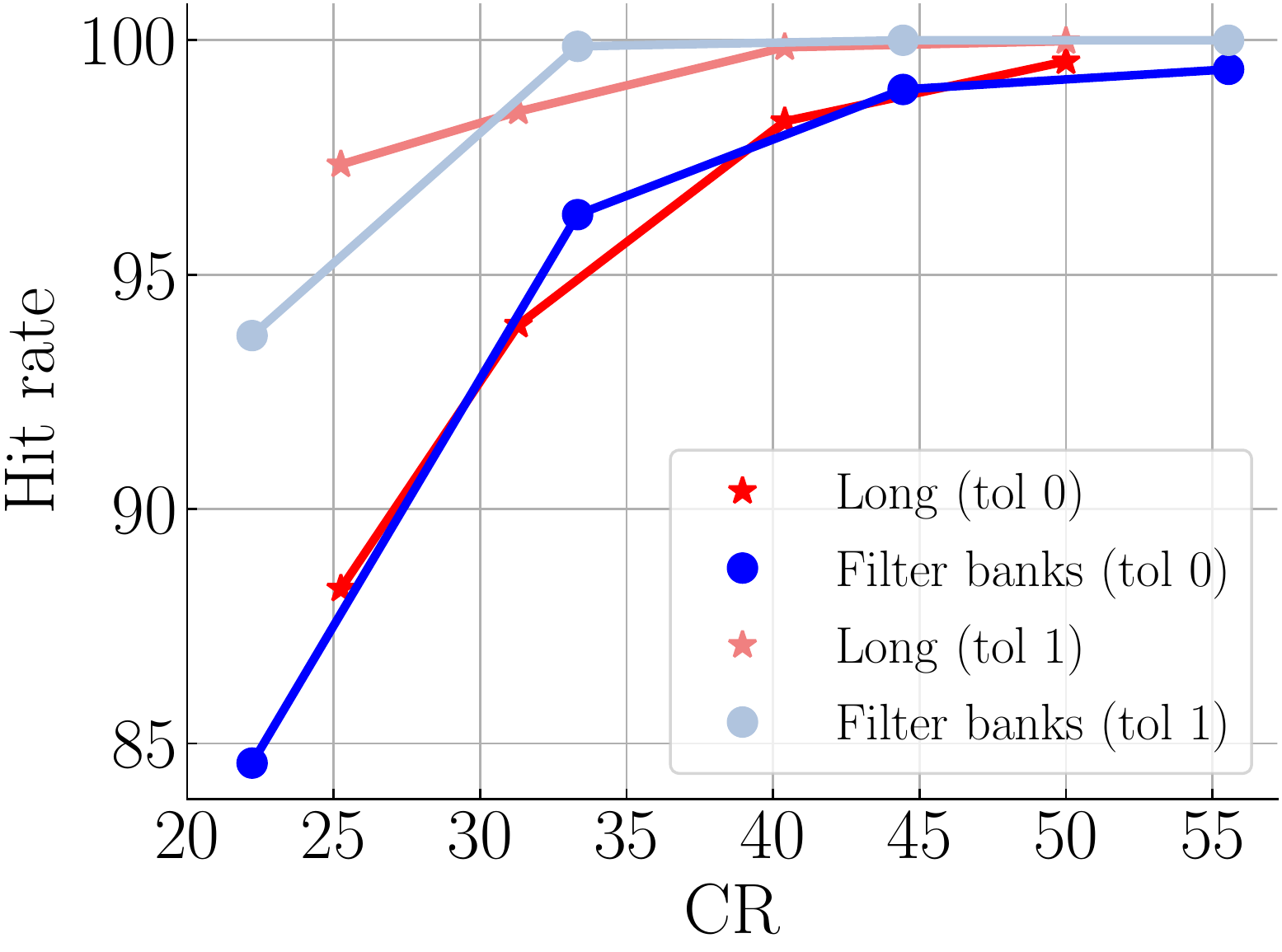}
\caption{Code hit rate}
\label{fig:bank_hitrate}
\end{subfigure}
\caption{Long and filter bank compression filters show similar performance.}
\label{fig:bank}
\end{figure}

\begin{table}[!t]
  \caption{Memory storage of the compression operator when it is structure long or filter banks. Memory cost is quantified by the dimensionality of the compression operator.}
   \centering
  \begin{tabular}{|c|c||c|c|}
    \hline
    \multicolumn{2}{|c||}{Long} & \multicolumn{2}{c|}{Filter banks}\\
    \hline
    CR & Memory Cost & CR &  Memory Cost\\
   \hline
    50 & 296 & 55.55 & 5 $\times$ 9 = 45 \\
    \hline
    40.4 & 277 & 44.44  & 4 $\times$ 9 = 36\\
     \hline
    31.31 & 259 & 33.33   & 3 $\times$ 9 = 27\\ 
     \hline
    25.25 & 247 & 22.22  & 2 $\times$ 9 = 18\\ 
    \hline
  \end{tabular}
  \label{tab:memory}
\end{table}

\begin{figure}[!t]
\centering
\begin{subfigure}[t]{0.49\linewidth}
\centering
\includegraphics[width=0.999\linewidth]{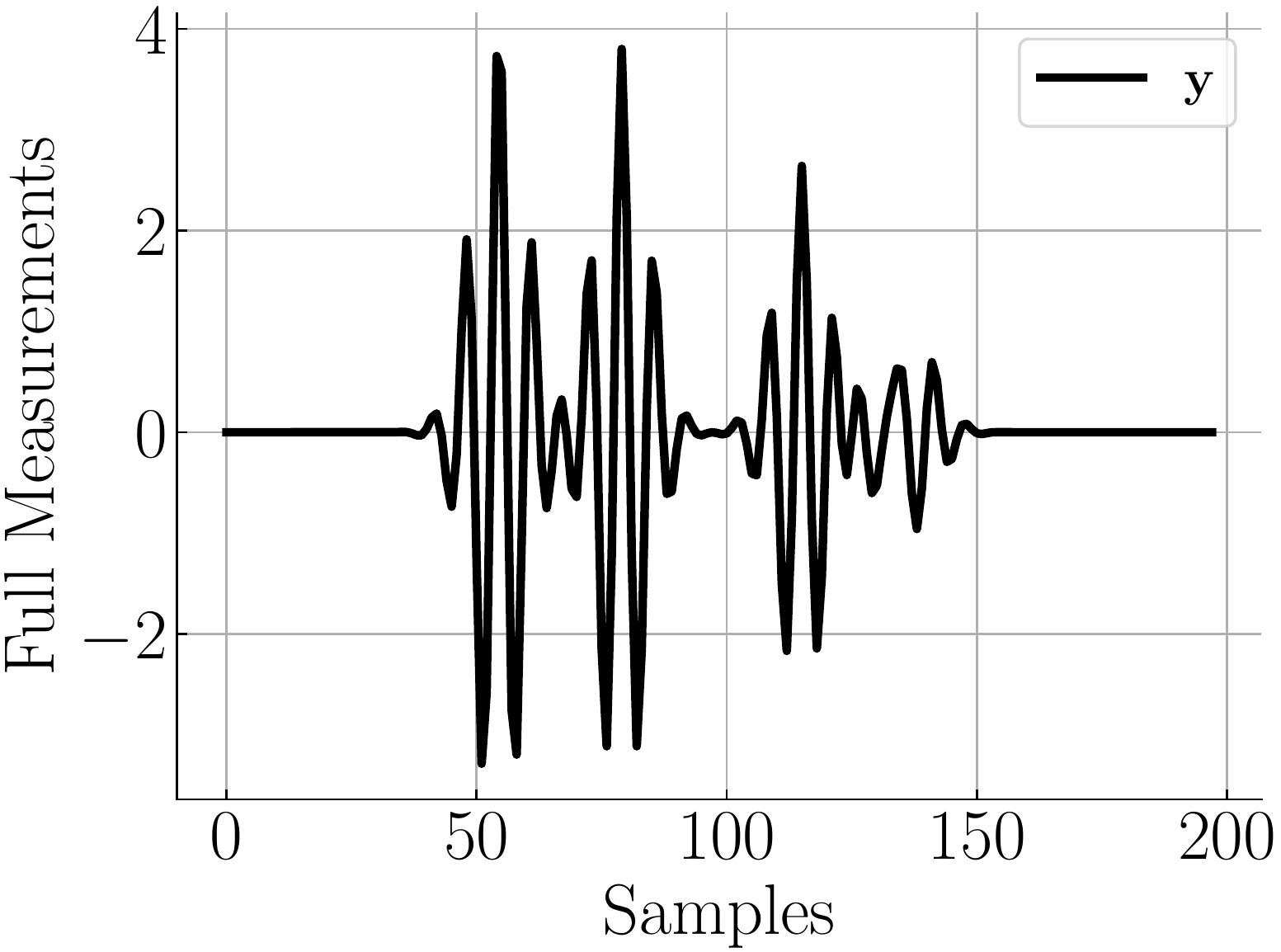}
\caption{Full measurements.}
\label{fig:data_vis_y}
\end{subfigure}
\begin{subfigure}[t]{0.49\linewidth}
\centering
\includegraphics[width=0.999\linewidth]{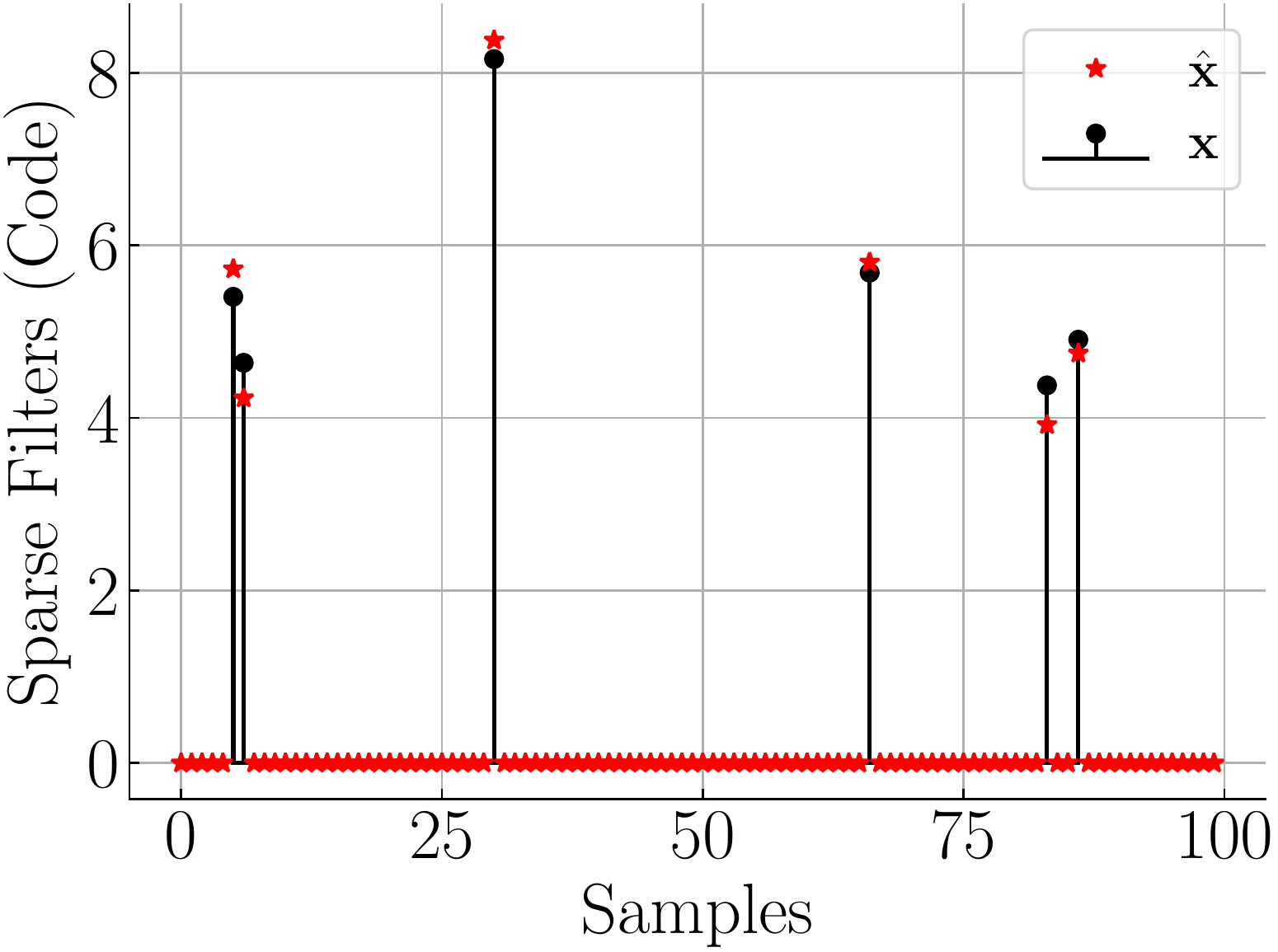}
\caption{Ground-truth and estimated sparse filters (codes).}
\label{fig:data_vis_x}
\end{subfigure}
\begin{subfigure}[t]{0.49\linewidth}
\centering
\includegraphics[width=0.999\linewidth]{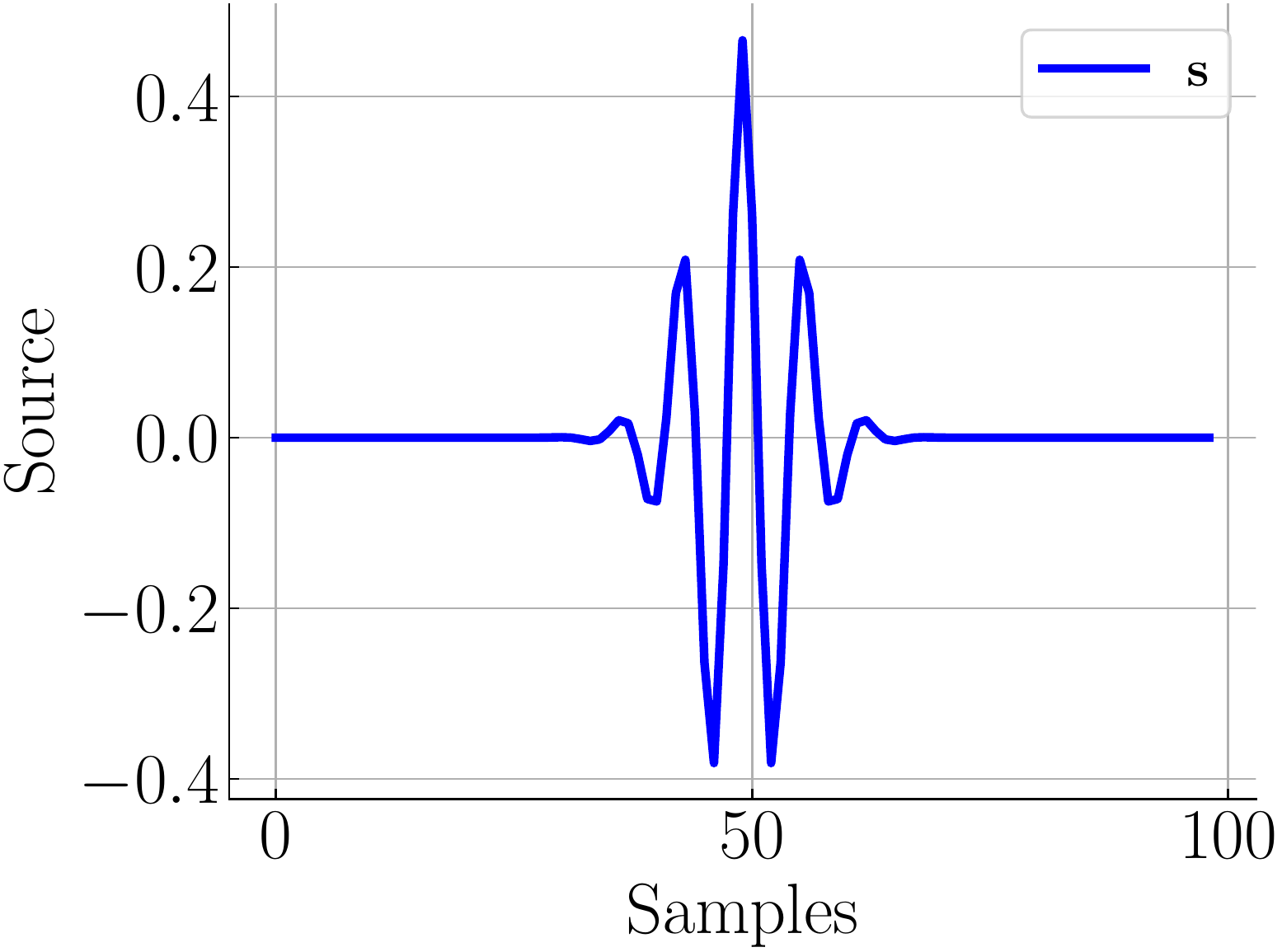}
\caption{Source shape.}
\label{fig:data_vis_s}
\end{subfigure}
\begin{subfigure}[t]{0.49\linewidth}
\centering
\includegraphics[width=0.999\linewidth]{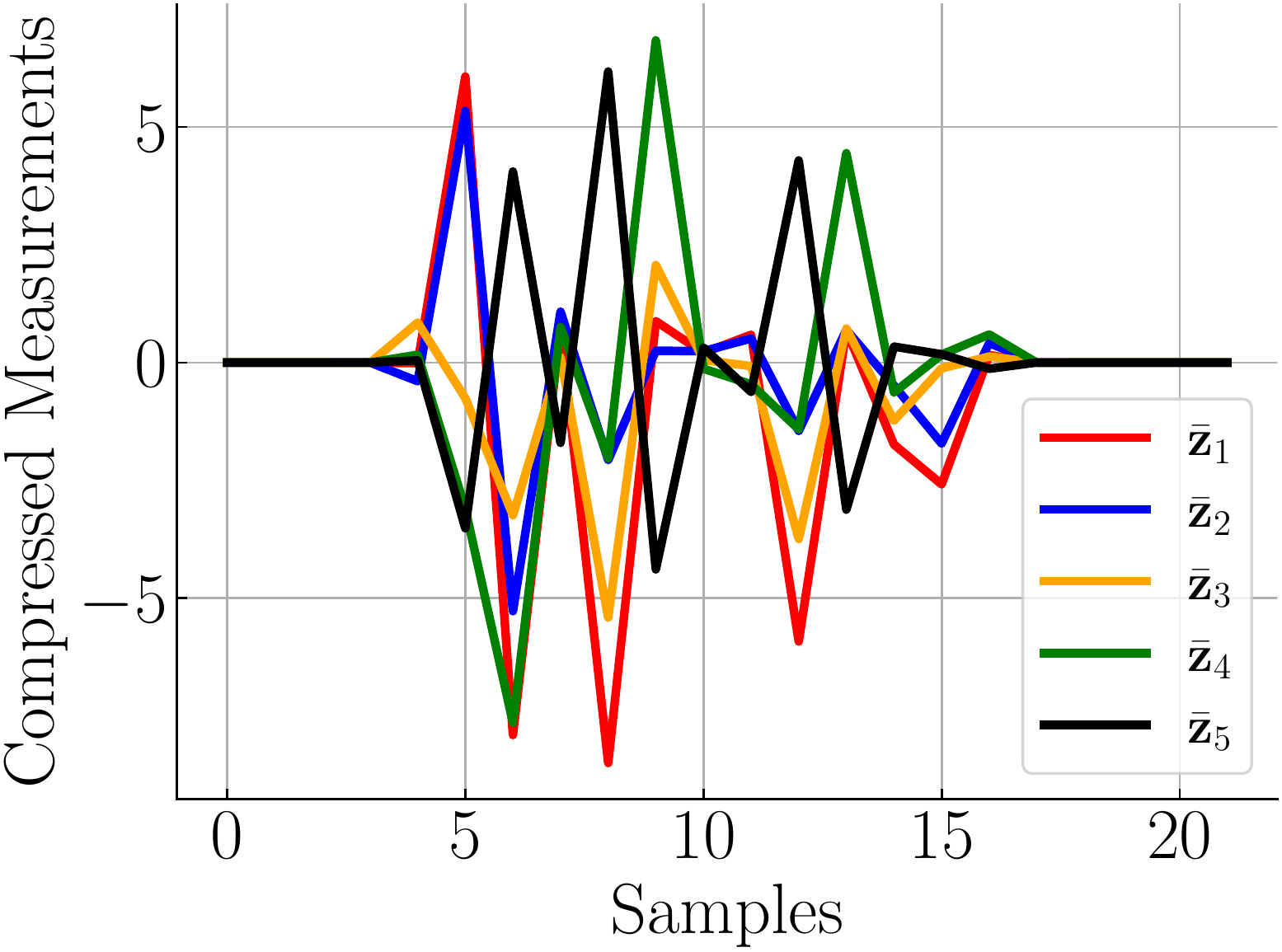}
\caption{Compressed data.}
\label{fig:data_vis_z}
\end{subfigure}
\begin{subfigure}[t]{0.49\linewidth}
\centering
\includegraphics[width=0.999\linewidth]{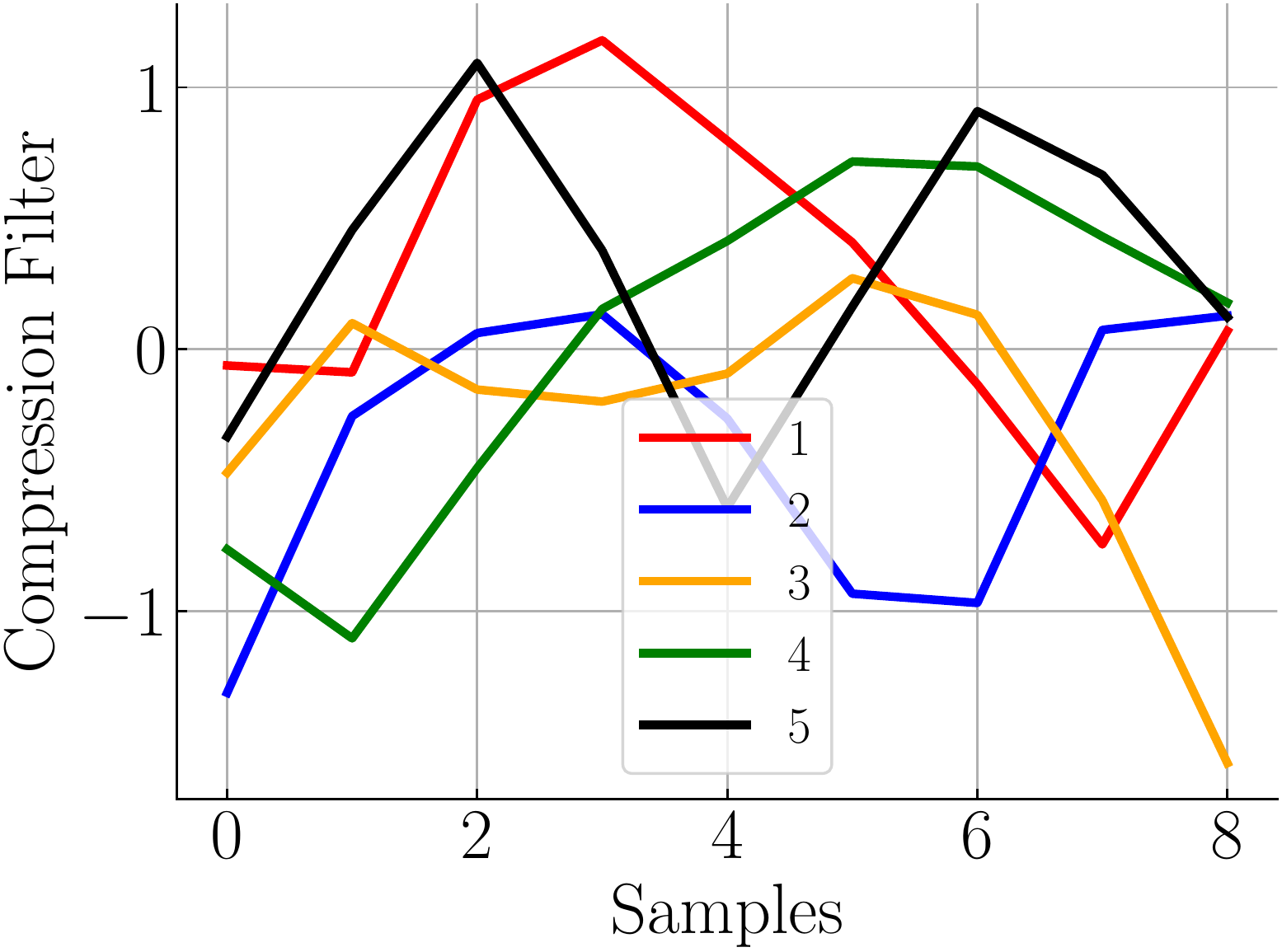}
\caption{Learned compression filters.}
\label{fig:data_vis_phi}
\end{subfigure}
\caption{Data visualizations on LS-MBD-LCISTA when the compression method of filter bank is used. The source has a wavelet shape, and the compression ratio is $\%55.55$.}
\label{fig:data_vis}
\end{figure}

\subsection{Toward data-driven unrolled denoising networks}


This section characterizes the performance of the proposed method in the presence of noise. The generated data contains the Gaussian source shape and the target amplitudes are following Uniform distribution $\text{Unif}(4,5)$. We trained the networks using $10{,}000$ training examples. We considered the case where the data is corrupted with additive Gaussian noise with $\sigma = 0.161$ (i.e., SNR of approximately $15$ dB). We used LS-MBD-LCISTA-freelayers-denoiser with $T=10$ unrolled layers and DnCNN~\cite{zhang2017beyond} as denoiser. The denoiser has a kernel size of $3$, $64$ feature maps, and a depth of $4$. We use validation and test sets of size $N=100$ where each example has its noise realization. The reported results are from the test set. For example, the reported results for the hit rate are averaged over $100$ examples. \cref{fig:noisy} shows that adding a denoiser to the unrolled network results in an improved performance both in terms of NRMSE and exact hit rate.


For when SNR is $15$ dB, we studied the effect of unrolling on denoising and target recovery. We trained LS-MBD-LCISTA-denoiser when $T=2$, $5$, and $10$. \Cref{fig:noisy_effect_of_unrolling} demonstrates that increasing the unrolling layers improves the performance.

\begin{figure}[!t]
\centering
\begin{subfigure}[t]{0.49\linewidth}
\centering
\includegraphics[width=0.999\linewidth]{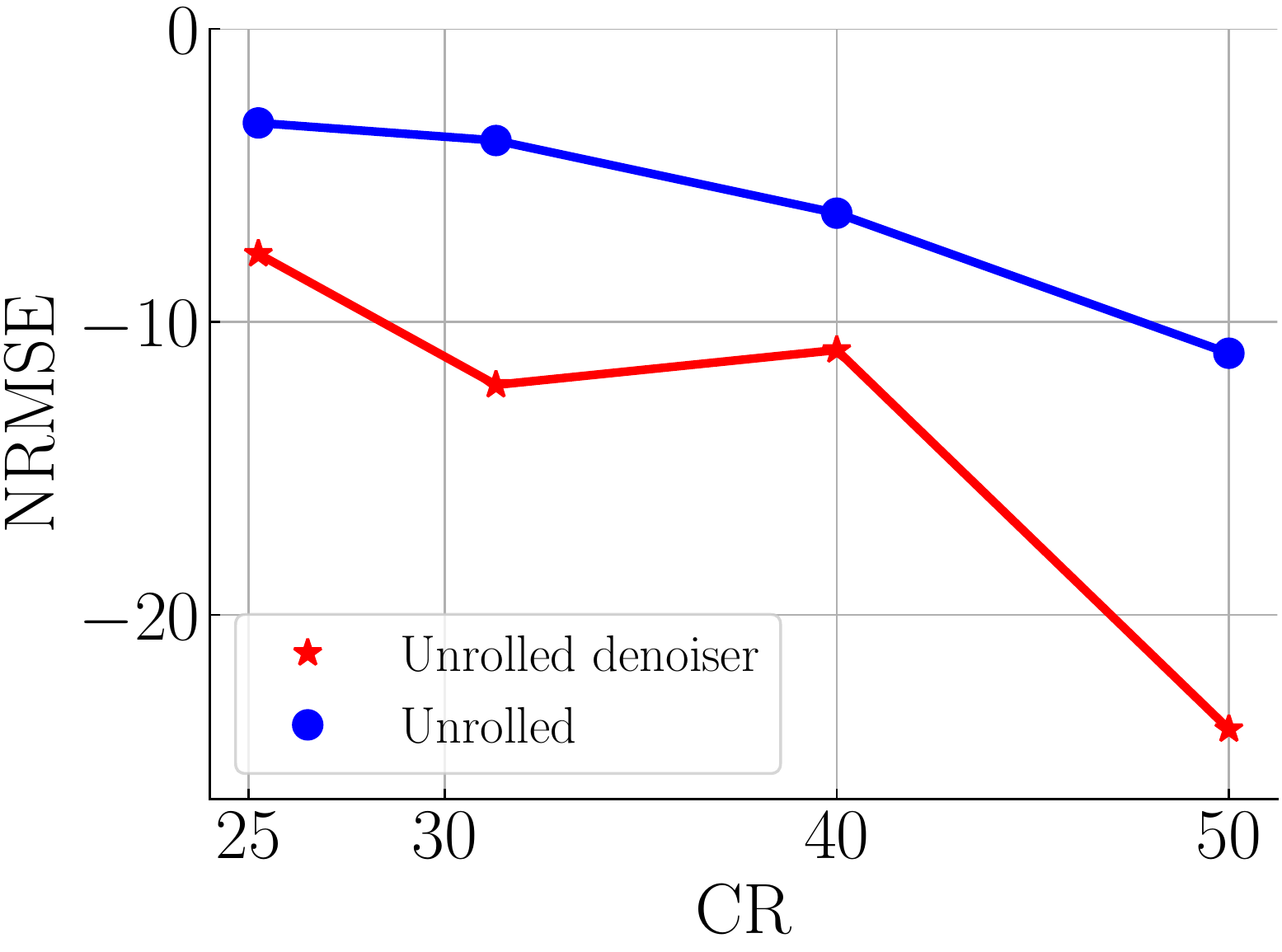}
\caption{Code NRMSE (SNR 15 dB).}
\label{fig:noisy_rmse_15}
\end{subfigure}
\begin{subfigure}[t]{0.49\linewidth}
\centering
\includegraphics[width=0.999\linewidth]{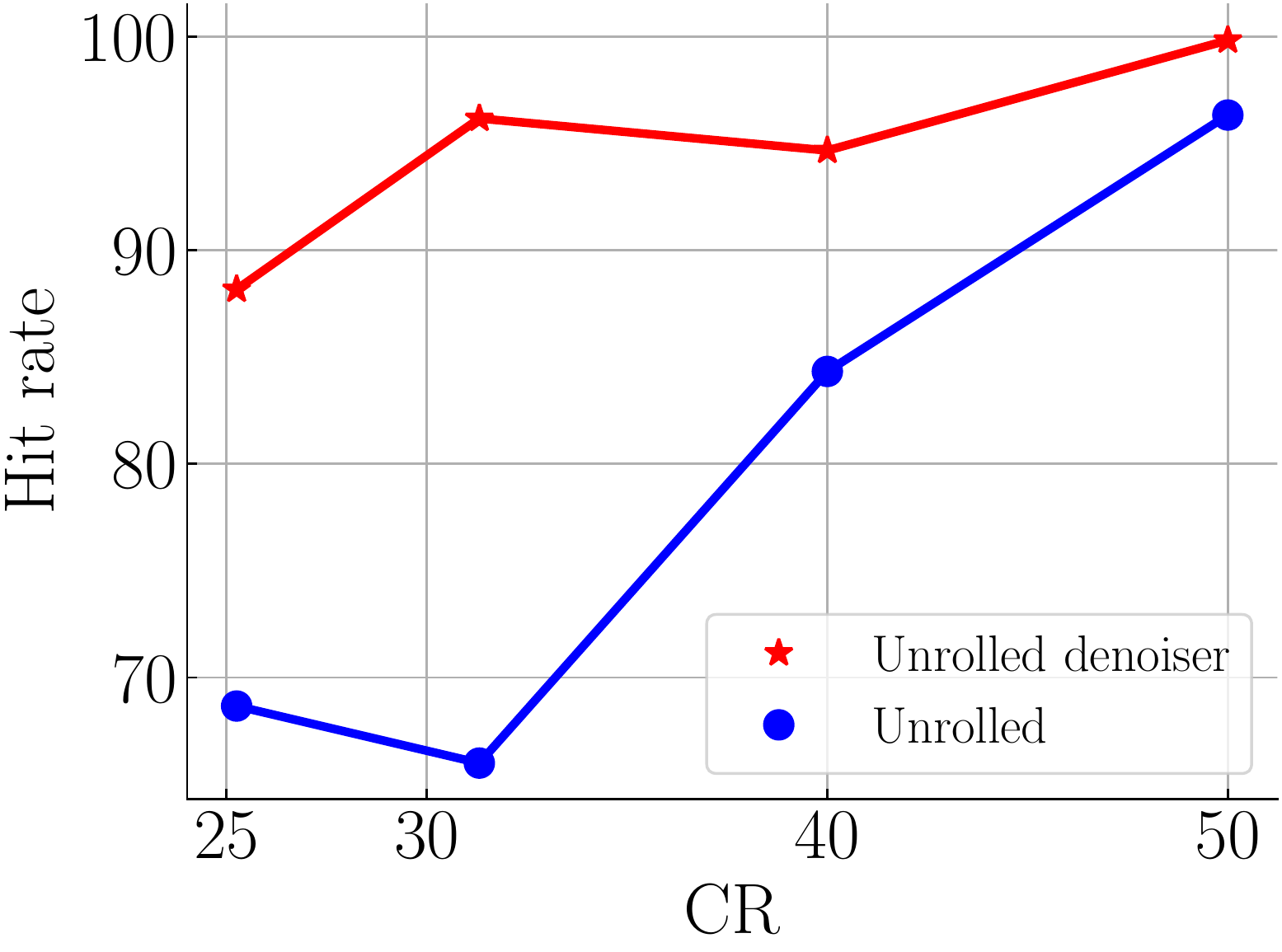}
\caption{Code hit rate (SNR 15 dB).}
\label{fig:noisy_hitrate_15}
\end{subfigure}
\begin{subfigure}[t]{0.49\linewidth}
\centering
\includegraphics[width=0.999\linewidth]{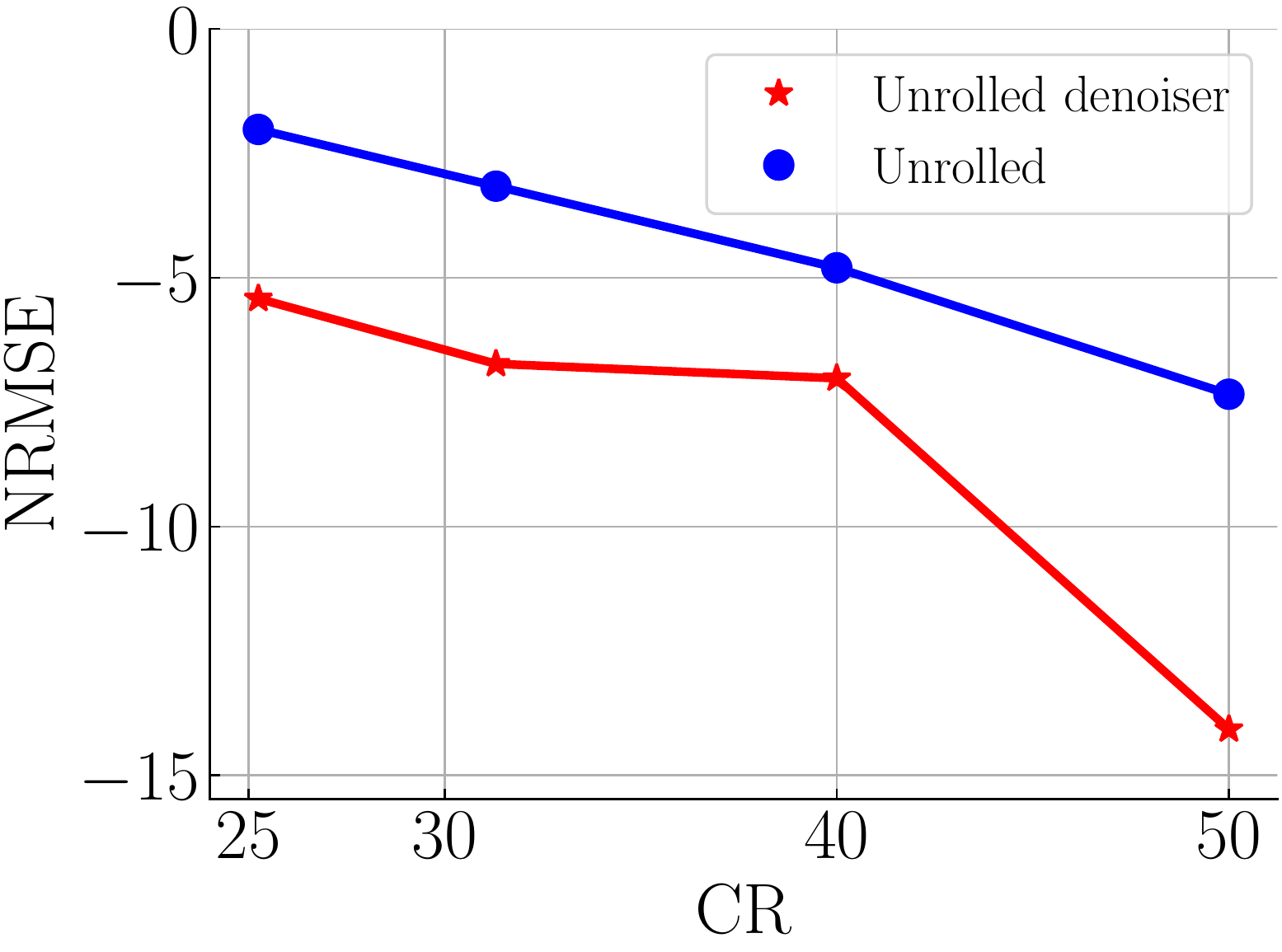}
\caption{Code NRMSE (SNR 10 dB).}
\label{fig:noisy_rmse_10}
\end{subfigure}
\begin{subfigure}[t]{0.49\linewidth}
\centering
\includegraphics[width=0.999\linewidth]{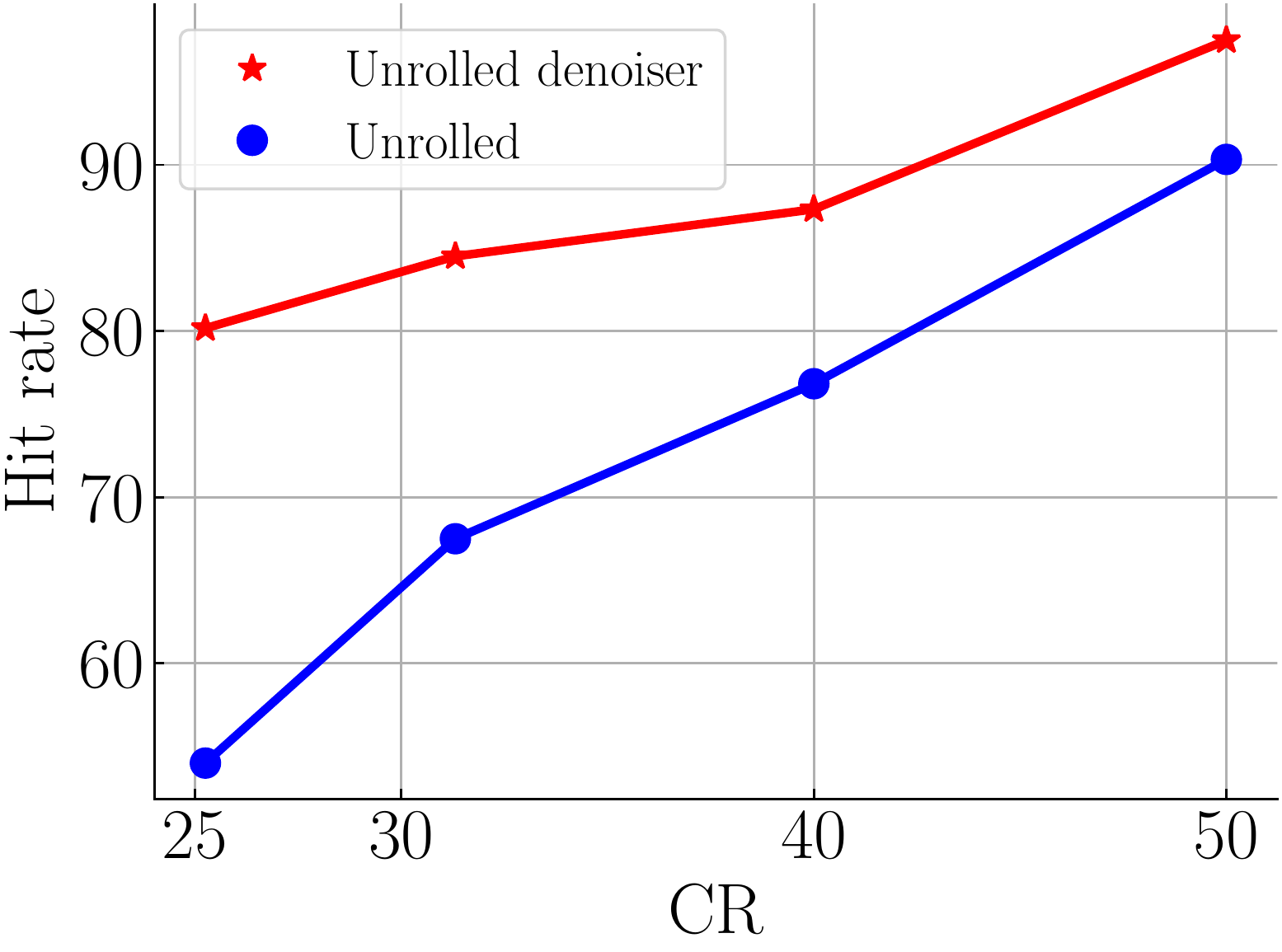}
\caption{Code hit rate (SNR 10 dB).}
\label{fig:noisy_hitrate_10}
\end{subfigure}
\caption{Outperformance of LS-MBD denoiser against LS-MBD in the presence of additive Gaussian noise. For SNR of $15$ and $10$ dB, $\sigma = 0.161$ and $0.297$, respectively. When SNR is $10$ dB, the variance of the reported average hit rate for LS-MBD-LCISTA is $1.25$, $2.52$, $3.39$, and $3.21$ for a compression ratio of $50$, $40$, $31.31$, and $25.25$, respectively. For the same scenario, the variance over the reported hit rate of LS-MBD-LCISTA-denoiser is $0.30$, $3.05$, $2.76$, and $2.49$. Moreover, when SNR is $15$ dB, the variance of the reported average hit rate for LS-MBD-LCISTA is $0.65$, $3.13$, $2.4$, $3.49$ for a compression ratio of $50$, $40$, $31.31$, and $25.25$, respectively. For the same scenario, The variance over the reported hit rate of LS-MBD-LCISTA-denoiser is $0.30$, $0.96$, $0.73$, and $2.32$.}
\label{fig:noisy}
\end{figure}

\begin{figure}[!t]
\centering
\begin{subfigure}[t]{0.49\linewidth}
\centering
\includegraphics[width=0.999\linewidth]{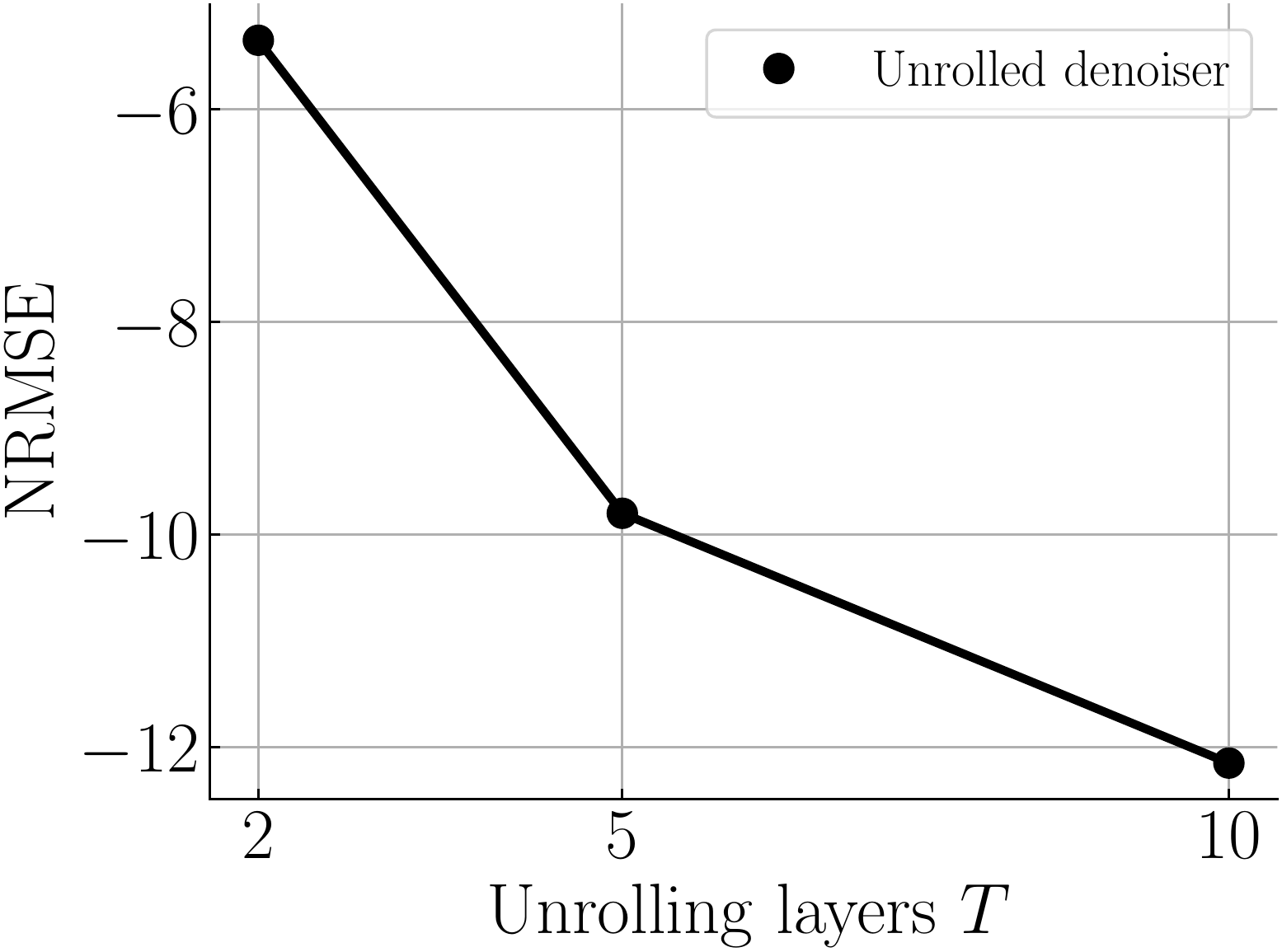}
\caption{Code NRMSE.}
\label{fig:noisy_rmse_15_effect_unrolling}
\end{subfigure}
\begin{subfigure}[t]{0.49\linewidth}
\centering
\includegraphics[width=0.999\linewidth]{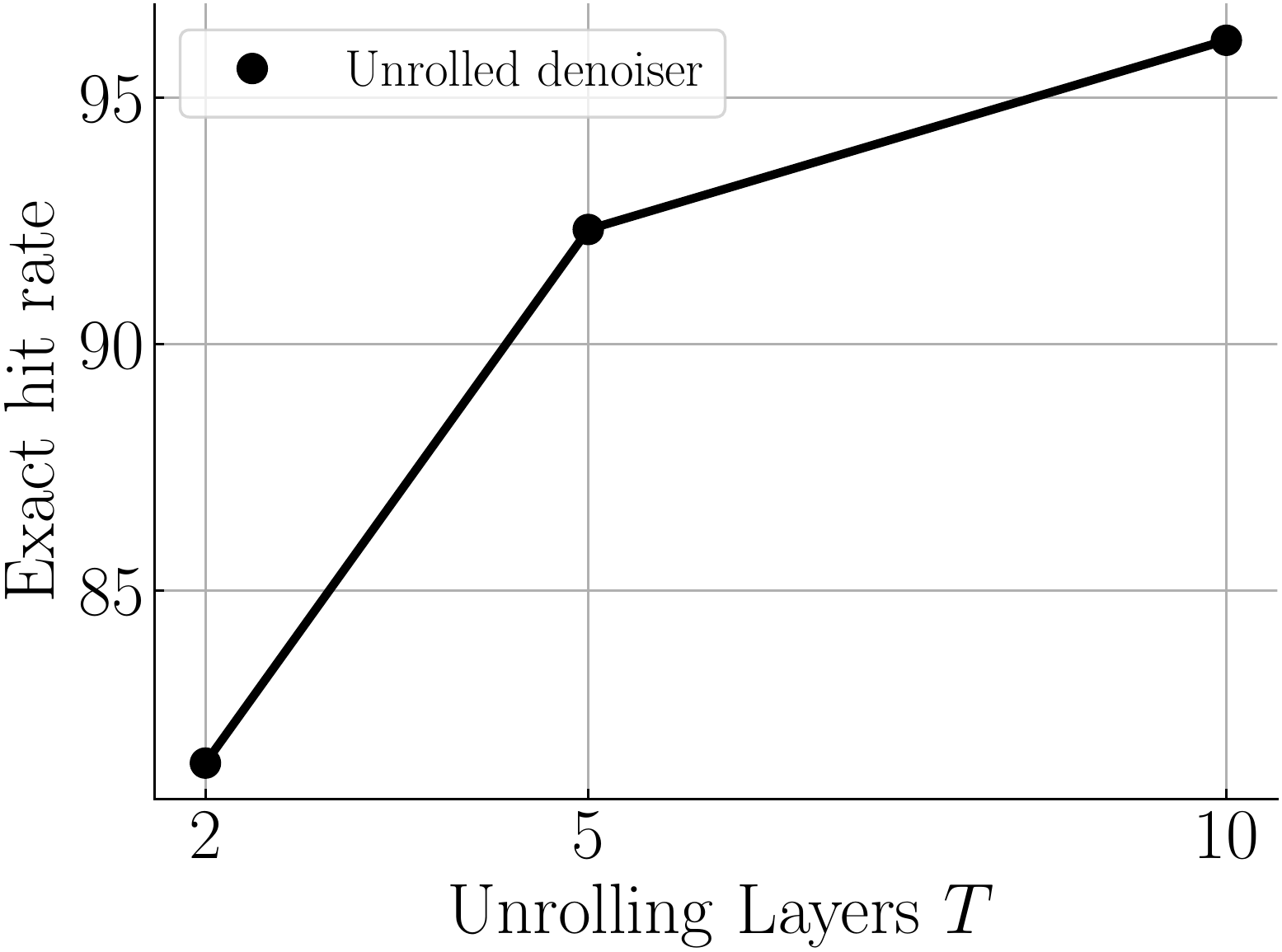}
\caption{Exact code hit rate.}
\label{fig:noisy_exact_hitrate_15_effect_unrolling}
\end{subfigure}
\caption{Effect of unrolling layers on the performance of LS-MBD denoiser in the presence of additive Gaussian noise. For SNR of $15$, with $\sigma = 0.161$.}
\label{fig:noisy_effect_of_unrolling}
\end{figure}

\subsection{Fewshot learning}\label{sec:fewshot}


In this section, we show the advantage of our unrolled learning (\cref{fig:infmap_unroll}) approach to conventional deep neural networks (\cref{fig:infmap_nn}) such as DnCNN~\cite{zhang2017beyond} in fewshot learning. We let $\text{CR} = 55.5$ and choose the source shape as a symmetric wavelet, which we assume is known. We utilized the filter banks compression (i.e., using $5$ compression filters with a kernel size of $9 \times 9$ and strides of $9$). For compressed DnCNN, after $\proj^{\text{T}}$, there is an additional fully-connected layer to go from $M_y$ to $M_x$ before the DnCNN network which its input and output have the same dimension. DnCNN uses kernel sizes of $3$, and $64$ feature maps. We set the depth of DnCNN to only $3$ layers lead to $32{,}645$ trainable parameters compared to $442$ of unrolling method; this is to alleviate the overfitting of DnCNN in low data regime.

As the amount of training data decreases from $10{,}000$ to $100$, the performance of compressed DnCNN significantly degrades, as opposed to our unrolling approach of LS-MBD-LCISTA (\cref{fig:fewshot}). We observed that unrolling approach outperforms DnCNN for all ranges of data. Most importantly, in a low data regime, DnCNN breaks down and reaches $0$ dB NRMSE and an exact hit rate below $\%40$ when there are only $100$ training examples.

\begin{figure}[!t]
\centering
\begin{subfigure}[t]{0.49\linewidth}
\centering
\includegraphics[width=0.999\linewidth]{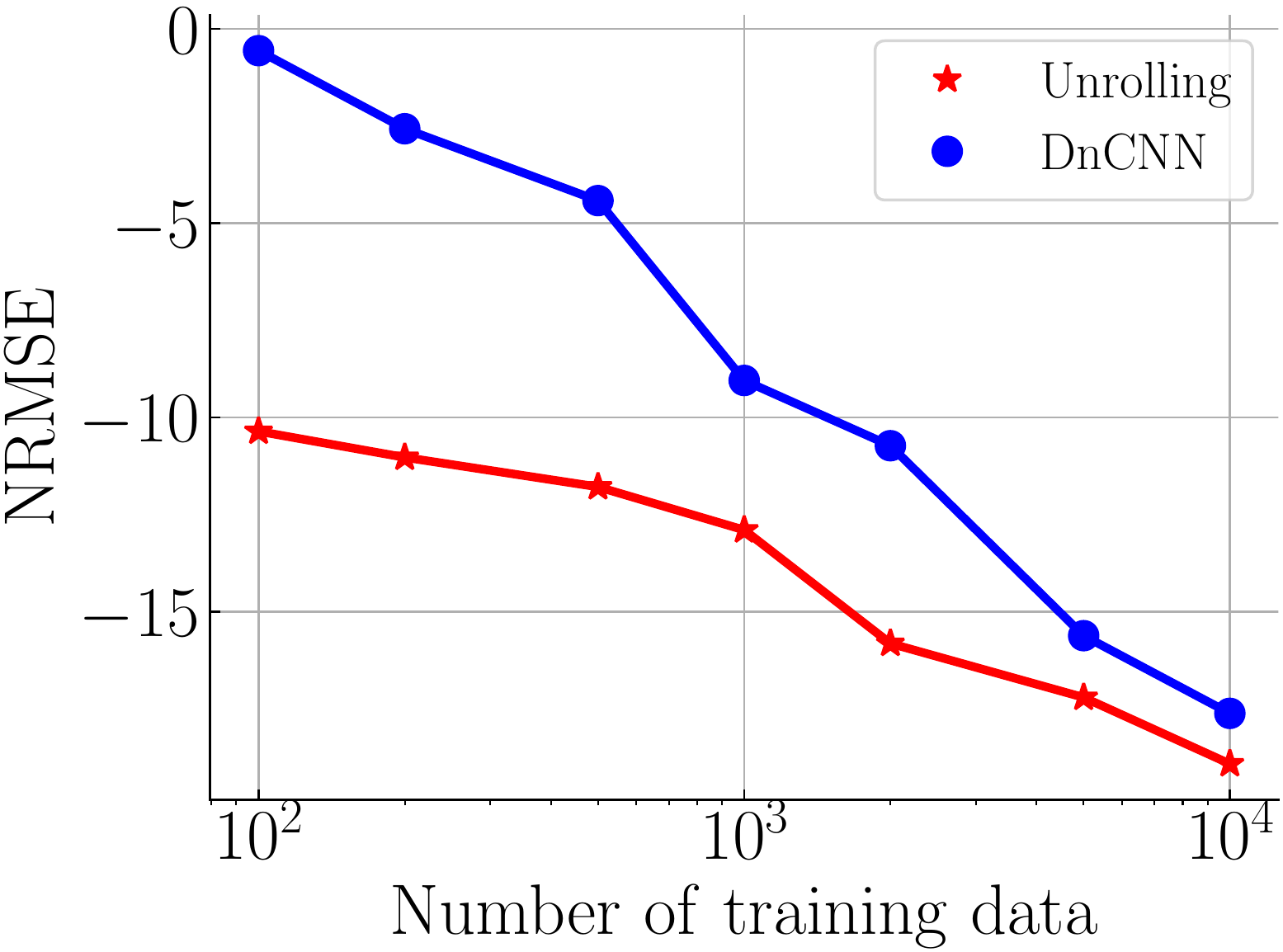}
\caption{Code NRMSE}
\label{fig:fewshot_rmse}
\end{subfigure}
\begin{subfigure}[t]{0.49\linewidth}
\centering
\includegraphics[width=0.999\linewidth]{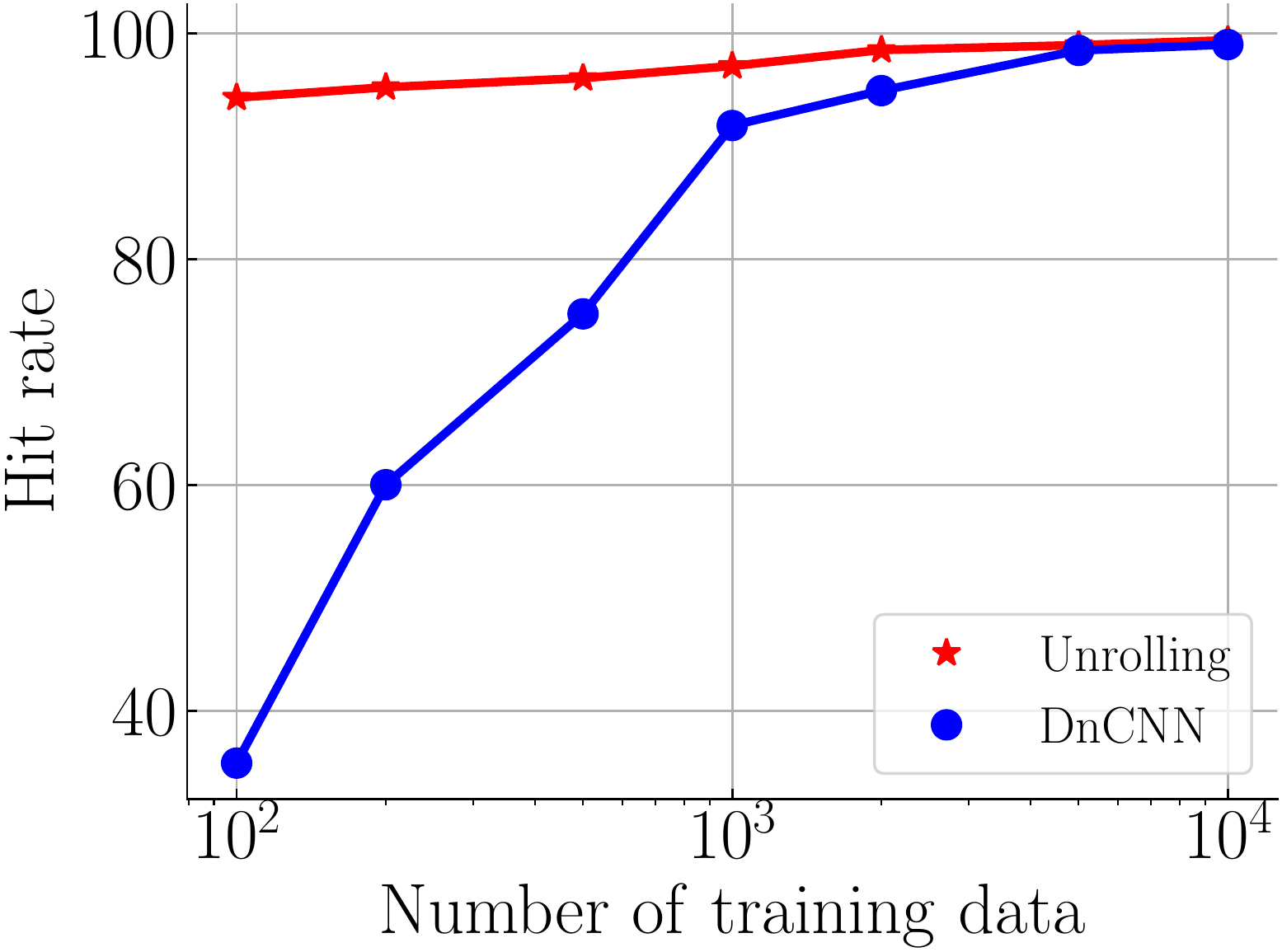}
\caption{Code exact hit rate}
\label{fig:fewshot_hitrate}
\end{subfigure}
\caption{Fewshot learning. Unrolling significantly outperforms conventional neural networks DnCNN in a data-limited regime.}
\label{fig:fewshot}
\end{figure}

\section{Conclusion}\label{sec:conclusion}

We studied the problem of sparse multichannel blind deconvolution (S-MBD). We proposed a filter-based compression to enable recovery from a small number of measurements. We theoretically studied the identifiability problem and provide conditions under which recovery is possible. We derived the necessary and sufficient conditions on the compression filter and the source for unique recovery. Through an unrolled learning framework, we proposed to learn the compression filter, resulting in better recovery from higher compression rates. Specifically, we designed a structured inference network based on an $\ell_1$-based iterative sparse coding algorithm to map compressed measurements to the sparse filters of interest. We demonstrated the robustness and superior performance of our proposed framework compared to traditional optimization-based methods. Moreover, we demonstrated the generalization capability of unrolled deep learning compared to generic deep learning in applications with data-limited regimes.

\appendix
In this section, we show that the D-SoS filter together with a source satisfying non-vanishing spectral properties \cite{mulleti_mbd} will follow the spark properties of~\Cref{theorem1}. The impulse response of the D-SoS filter considered in \cite{mulleti_mbd} is given as
\begin{align}
\mathbf{h}[m] = \sum_{k\in \mathcal{K}} e^{\mathrm{j}k\omega_0 m}, \quad 0 \leq m \leq M_h-1,
\label{eq:sos}
\end{align}
where $|\mathcal{K}| = 2L^2$ for sufficiency.
With the D-SoS filter, the entries of the matrix $\mathbf{A}_{(M_y+M_x-1, 2M_x-1)}$ for any source $\mathbf{s}_0 \in \mathbb{R}^{M_{s_0}}$ with $M_{s_0} < M_h$ are given as 
\begin{align}
    \mathbf{A}[m,l] = \sum_{k\in \mathcal{K}} S_0(e^{\mathrm{j}k\omega_0})\, e^{\mathrm{j}k\omega_0 (M_y+M_x+m-l-2)}, 
    \label{eq:D-SoS}
\end{align}
where $S_0(e^{\mathrm{j}\omega})$ is the discrete-time Fourier transform of $\mathbf{s}_0$ and $0 \leq m \leq M_h - M_y - M_x +1,\,\, 0\leq l \leq 2M_x-1$. The non-vanishing spectral condition states that $S_0(e^{\mathrm{j}\omega}) \neq 0$ for $k \in \mathcal{K}$. Then by assuming that $M_h \geq M_y+M_x+|\mathcal{K}|$, the matrix $\mathbf{A}(M_y+M_x-1, 2M_x-1)$ can be decomposed as
\begin{align}
    \mathbf{A}(M_y+M_x-1, 2M_x-1) = \mathbf{V}_1 \text{diag}(\mathbf{d}\{S_0(k\omega_0)\}_{k \in \mathcal{K}}) \mathbf{V}_2^*,
\end{align}
where entries of $|\mathcal{K}|$ length vector are consists of elements in the set $\{S_0(k\omega_0)e^{\mathrm{j}k\omega_0 (M_y+M_x-2)} \}_{k \in \mathcal{K}}$. The matrices $\mathbf{V}_1 \in \mathbb{C}^{M_y-M_x-1 \times |\mathcal{K}|}$ and $\mathbf{V}_2 \in \mathbb{C}^{2M_x-1 \times |\mathcal{K}|}$ are Vandermonde and with their $(m,k)$-th element is $e^{\mathrm{j}km\omega_0}$. As we discussed in the proof of~\Cref{theorem1}, the number of rows of $\mathbf{A}(M_y+M_x-1, 2M_x-1)$ should be greater than or equal to $2L^2$ for the spark condition to hold in the sufficient part. Hence, the matrix $\mathbf{V}_1$ is full column rank. In \cite{mulleti_mbd}, it was assumed that $S_0(k\omega_0) \neq 0, k \in \mathcal{K}$ which ensures that the matrix $\mathbf{V}_1\text{diag}(\mathbf{d}\{S_0(k\omega_0)\}_{k \in \mathcal{K}})$ is full rank. Hence, spark of matrix $\mathbf{A}(M_y+M_x-1, 2M_x-1)$ should be equal to that of matrix $\mathbf{V}_2^*$. Given that $\omega_0$ and $\mathcal{K}$ form universal sampling pattern \cite{mulleti_mbd} an $|\mathcal{K}|$ columns of matrix $\mathbf{V}_2^*$ are linearly independent and hence the matrix has full spark \cite{mulleti_mbd}. Hence, $\text{Spark}(\mathbf{A}(M_y+M_x-1, 2M_x-1)) >2L^2$.

\bibliographystyle{IEEEtran}
\bibliography{IEEETSP_v7}
	
\end{document}